\documentclass[letterpaper,11pt]{article}
\usepackage{amsmath, amsthm, amssymb}
\usepackage[usenames, dvipsnames]{color}
\usepackage[normalem]{ulem} %for strikethrough
\usepackage{fullpage}
\usepackage[numbers, sort]{natbib}
\usepackage[noend]{algorithmic}
\usepackage{enumerate}
\usepackage{hyperref}
\hypersetup{
  hidelinks,
  colorlinks=true,
  allcolors=black,
  pdfstartview=Fit,
  breaklinks=true
}
\usepackage{xspace,xcolor}
\usepackage{placeins}
\usepackage{graphicx}
\usepackage{caption}
\usepackage{subcaption}
\usepackage{enumitem, linegoal}
\usepackage[ruled,vlined]{algorithm2e}
\usepackage{comment}

\newtheorem{observation}{Observation}[section]

\usepackage{mathtools}
\usepackage[flushleft]{threeparttable}
\usepackage{verbatim}
\usepackage{setspace}
\usepackage{bbm}
\usepackage{thm-restate}
\usepackage{footnote}
\makesavenoteenv{tabular}
\makesavenoteenv{table}
\usepackage{authblk}

\usepackage[margin=1in]{geometry}

\definecolor{crimsonglory}{rgb}{0,0,0}%{0.75, 0.0, 0.2}

\newtheorem{example}{Example}

 \newtheorem{theorem}{Theorem}[section]
 
 \newtheorem{lemma}[theorem]{Lemma}

 \newtheorem{definition}[theorem]{Definition}

\makeatletter
\def\GrabProofArgument[#1]{ #1: \egroup\ignorespaces}
\def\proof{\noindent\textbf\bgroup Proof%
	\@ifnextchar[{\GrabProofArgument}{. \egroup\ignorespaces}}

\makeatother

\usepackage[latin1]{inputenc}
\usepackage{tikz}
\usetikzlibrary{positioning,shapes.geometric,calc}

\newcommand{\flower}{snowflake\xspace}
\newcommand{\Flower}{Snowflake\xspace}
\newcommand{\pistil}{center\xspace}

\newcommand{\fast}{\textit{$k$-broadcasting}\xspace}
\newcommand{\Fast}{\textit{$k$-broadcasting}\xspace}
\newcommand{\twofast}{\textit{$2$-broadcasting}\xspace}

\newcommand{\todo}[1] {}
\newcommand{\todoaida}[1]{}
\newcommand{\todoseyed}[1] {}
\newcommand{\tododone}[1] {}
\newcommand{\seyed}[1]{}
\newcommand{\aida}[1]{}
\newcommand{\todosumedha}[1]{}
\newcommand{\toask}[1]{}

\newcommand{\br}{\boldsymbol{br}}
\newcommand{\bropt}{{\boldsymbol{br}^*}}
\newcommand{\oh}{\mathcal{O}}
\newcommand{\bigell}{L}

\newcommand{\telebr}{\textsc{Telephone Broadcasting}\xspace}
\newcommand{\sat}{\textsc{SAT}\xspace}
\newcommand{\sattt}{\textsc{$3$-SAT}\xspace}
\newcommand{\sssattt}{\textsc{$3,4$-SAT}\xspace}

\newcommand{\tis}{\textsc{TwIS}\xspace}
\newcommand{\tislong}{\textsc{Twin Interval Selection}\xspace}
\newcommand{\dspr}{\textsc{DoSePr}\xspace}
\newcommand{\dsprlong}{\textsc{Dome Selection with Prefix Restrictions}\xspace}

\newcommand{\cater}{reduced caterpillar\xspace}

\newcommand{\cds}{\dspr}
\newcommand{\telebg}{\textsc{Telephone Broadguess}\xspace}

\newcommand{\singlefunc}[1]{\texttt{single-br(#1)}\xspace}
\newcommand{\doublefunc}[1]{\texttt{double-br(#1)}\xspace}

\newcommand{\singlefuncc}{\texttt{single-br}\xspace}
\newcommand{\doublefuncc}{\texttt{double-br}\xspace}

\newcommand{\ouralgo}{\textsc{Cactus Broadcaster}\xspace}

\newcommand{\pwidth}{w}

\newcommand{\sspan}{\text{span}}

\newcommand{\rrank}{\text{rank}}

\newcommand{\figcomponent}{
    \begin{figure}
	\centering
	\includegraphics[width=0.3\linewidth]{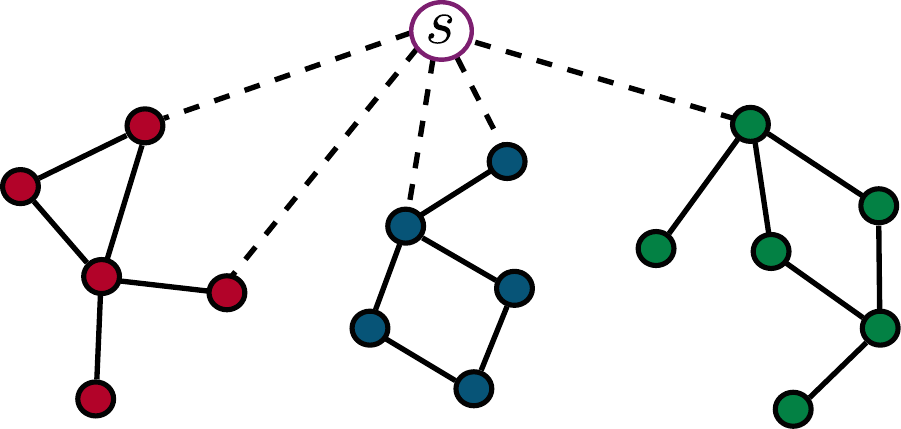}
	\caption{Two possible types for broadcasting to connected components after deleting $s$ in \singlefunc{$G,s$} are shown. The green vertices form a single-neighbor component. The red and blue components are double-neighbor components.}
    \label{fig:2approxSingleSource}

\end{figure}
}

\newcommand{\twoapproxnewfig}{

\begin{figure}
	\centering
	\includegraphics[width=0.46\linewidth]{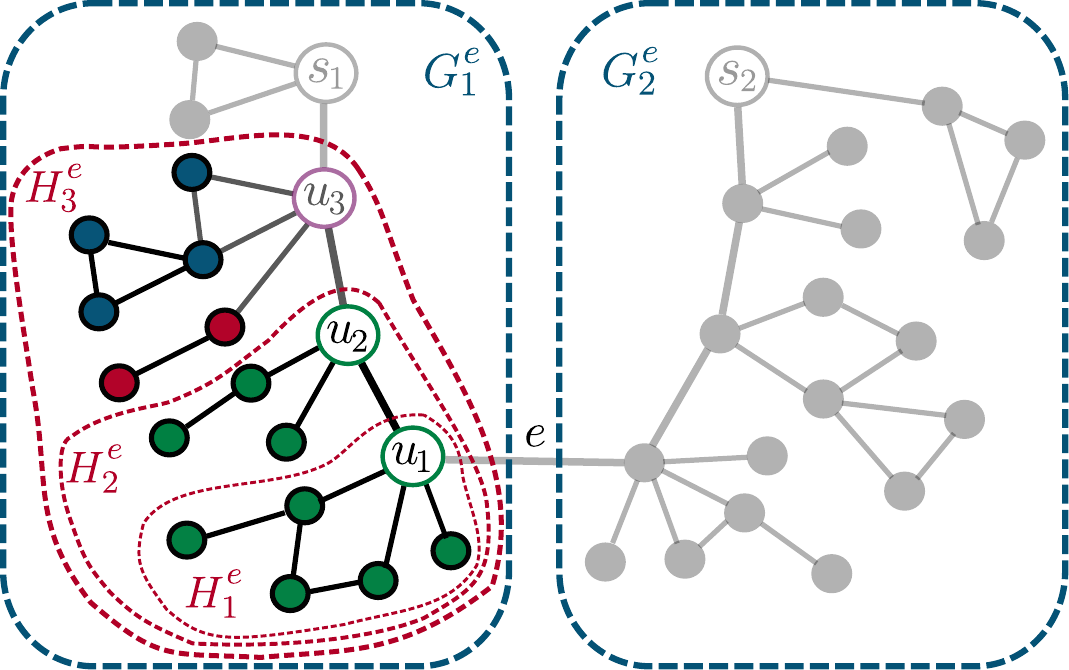}
	\caption{An illustration of Phase $2$ of the \doublefuncc method. Different components in iteration $i=3$ of Phase $2$ (for $u_3$) are highlighted in different colors. 
  }
    \label{fig:2approxDoubleMerge}
\end{figure}
}

\newcounter{proccnt}

\newcommand{\konote}[1]{}

\usepackage[normalem]{ulem} %for strikethrough

\title{On the Complexity of Telephone Broadcasting \\ \Large{From Cacti to Bounded Pathwidth Graphs}}
\author{Aida Aminian}
\author{Shahin Kamali}
\author{Seyed-Mohammad Seyed-Javadi}
\author{Sumedha}

\affil{York University, Toronto, Canada \\
\texttt{\{aminian, kamalis, smjavadi\}@yorku.ca}, \texttt{sumedhab@my.yorku.ca}}

\SetCommentSty{mycommfont}

\SetKwInput{KwInput}{Input}                % Set the Input
\SetKwInput{KwOutput}{Output}  
\begin{document}
	\newcommand{\ignore}[1]{}
\renewcommand{\theenumi}{(\roman{enumi}).}
\renewcommand{\labelenumi}{\theenumi}
\sloppy
\date{} %\date{\today}
\newenvironment{subproof}[1][\proofname]{%%
	  \renewcommand{\Box}{ \blacksquare}%
	\begin{proof}[#1]%
	}{%
	\end{proof}%
}
\newenvironment{subproof2}{%%
	\renewcommand{\Box}{ \blacksquare}%
	\begin{proof}%
	}{%
	\end{proof}%
}

% New commands added by Sumedha for proof subsections -> please lmk if you prefer to not have this used.
\newlist{pfparts}{description}{1}
\setlist[pfparts,1]{%
  % font=\normalfont\textsf,
  itemindent=2pt,
  wide,
  itemsep=0pt,topsep=2pt,
  labelsep=0.75ex
}
% From 
% https://tex.stackexchange.com/questions/532909/labelling-a-section-or-subproof-within-a-proof

% \makeatletter
% \def\namedlabel#1#2{\begingroup
%     #2%
%     \def\@currentlabel{#2}%
%     \phantomsection\label{#1}\endgroup
% }
% \makeatother

\maketitle

\thispagestyle{empty}
\allowdisplaybreaks

\begin{abstract}
In \telebr, the goal is to disseminate a message from a given source vertex of an input graph to all other vertices in the minimum number of rounds, where at each round, an informed vertex can send the message to at most one of its uninformed neighbors. For general graphs of $n$ vertices, the problem is NP-complete, and the best existing algorithm has an approximation factor of $\oh(\log n/ \log \log n)$. The existence of a constant factor approximation for the general graphs is still unknown. 

In this paper, we study the problem in two simple families of sparse graphs, namely, cacti and graphs of bounded pathwidth. There have been several efforts to understand the complexity of the problem in cactus graphs, mostly establishing the presence of polynomial-time solutions for restricted families of cactus graphs (e.g.,~\cite{vcevnik2017broadcasting, ehresmann2021approximation, harutyunyan2009necklace, Harutyunyan2007unicyclic, harutyunyan2008unicyclic, harutyunyan2023chainring}). Despite these efforts, the complexity of the problem in arbitrary cactus graphs remained open. We settle this question by establishing the NP-completeness of telephone broadcasting in cactus graphs. For that, we show the problem is NP-complete in a simple subfamily of cactus graphs, which we call \flower graphs. These graphs not only are cacti but also have pathwidth $2$. These results establish that, despite being polynomial-time solvable in trees, the problem becomes NP-complete in very simple extensions of trees.

On the positive side, we present constant-factor approximation algorithms for the studied families of graphs, namely, an algorithm with an approximation factor of $2$ for cactus graphs and an approximation factor of $\oh(1)$ for graphs of bounded pathwidth.
\end{abstract}
\section{Introduction}
The \telebr problem involves disseminating a message from a single given source vertex to all other vertices in a network through a series of \emph{telephone calls}.
The network is often modeled as an undirected and unweighted graph of $n$ vertices. Communication takes place in synchronous rounds. Initially, only the source is informed.
 During each round, any informed vertex can transmit the message to at most one of its uninformed neighbors via a ``call''. The goal is to minimize the number of rounds required to inform the entire network. \citet{hedetniemi1988broadsurvey} identifies \telebr as a fundamental primitive in distributed computing and communication theory, forming the basis for many advanced tasks in these fields. 
 
\citet{slater1981nptree} established the NP-completeness of \telebr. 
Nonetheless, efficient algorithms have been developed for specific classes of graphs. 
In particular, \citet{fraigniaud2002polynomial} demonstrated that the problem is solvable in polynomial time for trees. There are polynomial
algorithms for several other graph families; see, e.g., \cite{damaschke2024starclique, gholami2023fullytree, Liestman1988Boundeddegree, stohr1991butterfly}.

In this paper, we study cactus graphs, which are graphs in which

any two cycles share at most one vertex. 
Cacti are a natural generalization of trees and ring graphs, providing a flexible model for applications such as wireless sensor networks, particularly when tree structures are too restrictive~\cite{ben2012centdian}.
 
There have been several studies for broadcasting in specific families of cactus graphs \cite{vcevnik2017broadcasting, ehresmann2021approximation, harutyunyan2009necklace, Harutyunyan2007unicyclic, harutyunyan2008unicyclic, harutyunyan2023chainring}. For instance, for unicyclic graphs, a simple subset of cactus graphs containing exactly one cycle, the problem is solvable in linear time~\cite{Harutyunyan2007unicyclic}.
Similarly, \cite{harutyunyan2023chainring} proved that the chain of rings, which consists of cycles connected sequentially by a single vertex, has an optimal algorithm that runs in linear time. 
In $k$-restricted cactus graphs, where each vertex belongs to at most $k$ cycles, for a fixed constant $k$, \citet{vcevnik2017broadcasting} proposed algorithms that compute the optimal broadcast scheme 
in $\oh(n)$ time. Despite all these efforts, the complexity of \telebr for arbitrary cactus graphs has remained open, a question that we resolve in this paper. 

\telebr is NP-hard for general graphs~\cite{GareyJ79}, and even for restricted graphs. In particular, Tale recently showed that the problem remains NP-hard for graphs of pathwidth of $3$ \cite{tale2024double}, a result that naturally extends to graphs with higher pathwidths. However, the complexity of the problem for graphs of pathwidth $2$ has remained an open question, which we address in this paper.

\citet{elkin2002lowerbound} proved that approximating \telebr within a factor of $3 - \epsilon$ for any $\epsilon > 0$ is NP-hard. \citet{kortsarz1995approximation} showed that \telebr in general graphs has an approximation ratio of $\oh(\log n / \log \log n)$. 
It is possible that there is a constant factor approximation for general graphs. However, 
constant-factor approximation exists only for restricted graph classes such as unit disk graphs~\cite{shang2010unitdisk} and certain sub-families of cactus graphs such as $k$-cycle graphs~\cite{BhabakH15}.

\subsection{Contribution}
This paper investigates the complexity and approximation algorithms for the \telebr problem for cacti and graphs of bounded pathwidth. The key contributions are summarized as follows: \vspace{1mm}
\begin{itemize}
    \item We present a simple polynomial-time algorithm, named \ouralgo, and prove it has an approximation factor of 2 for broadcasting in cactus graphs (Theorem~\ref{thm:2approx}). Constant-factor approximations are known for certain subfamilies of cactus graphs (e.g., $k$-cycle graphs~\cite{BhabakH15}), and our result extends this to arbitrary cactus graphs.   
\ouralgo is  
reminiscent of the tree algorithm of~\cite{fraigniaud2002polynomial} and leverages the separatability of cactus graphs. 
In our analysis of \ouralgo, we use the \fast model of broadcasting~\cite{grigni1991tightkbroad} as a reference point, where each vertex can inform up to two neighbors in a single round. 

\item Our main contribution is to establish the NP-completeness of \telebr problem in ``\flower graphs'', which are a subclass of cactus graphs and also have pathwidth at most $2$, therefore resolving the complexity of the problem in these graph families (Theorem~\ref{thm:hardness-flower}). For a formal definition of \flower graphs, refer to Definition~\ref{def:ourflowergraph}. This hardness result is achieved through a series of reductions that start from \sssattt, a variant of the satisfiability problem that is NP-complete
\cite{tovey1984simplified}. 
    \item We show the existence of a constant-factor approximation for graphs of bounded pathwidth.  
    For that, we 
    show the algorithm of \citet{elkin2006sublogarithmic} has an approximation factor of 
$\oh(4^\pwidth)$ for any graph of constant pathwidth $\pwidth$ (Theorem~\ref{thm:approximation-constant}). Note that this result does not rely on having a path decomposition certifying a bounded pathwidth. 
Constant-factor approximation algorithms are known for certain families of graphs of bounded path-width such as 
$k$-path graphs, which admit a $2$-approximation algorithm~\cite{harutyunyan2023kpath}, and our result extends this to any graph of bounded pathwidth. 
\end{itemize}

\subsection{Paper Structure}
In Section~\ref{sec:preliminaries}, we present the preliminaries.
In Section~\ref{sec:2approx}, we propose \ouralgo, which gives a $2$-approximation 
for cactus graphs. In Section~\ref{sec:hardness}, we 
establish the NP-completeness of the problem in \flower graphs.
In Section~\ref{sec:pathwidth-approx}, 
we present a constant-factor approximation algorithm for graphs of bounded pathwidth. 
We conclude in Section~\ref{section:concluding}.

\section{Preliminaries} \label{sec:preliminaries}
For a positive integer $n$, we use notation $[n]$ to denote 
$\{1,2,\ldots, n\}$. Also, we use $[i,j]$ to refer to denote $\{i,i+1,\ldots, j\}$. 
For a graph $G$, we use $G\setminus \{v\}$ to refer to the subgraph of $G$ induced by all vertices of $G$ except $v$.
 \begin{definition}[\cite{hedetniemi1988broadsurvey}]
    An instance $(G,s)$ of the 
     \emph{\telebr} problem  
     is defined by a connected, undirected, and unweighted graph $G=(V, E)$ and a vertex $s \in V$, where $s$ is the only informed vertex. 
     The broadcasting protocol is synchronous and occurs in discrete rounds. In each round, an informed vertex can inform at most one of its uninformed neighbors. The goal is to broadcast the message as quickly as possible so that all vertices in $V$ get informed in the minimum number of rounds.
 \end{definition}

A broadcast scheme describes the ordering at which each vertex informs its neighbors. One can describe a broadcast scheme $S$ with a \emph{broadcast tree}, which is a spanning tree of $G$ rooted at source $s$; if a vertex $u$ is informed through vertex $v$, then $u$ will be a child of $v$ in the broadcast tree. Given that the optimal broadcast scheme of trees can be computed in linear time, a broadcast tree can fully describe the broadcast scheme. We use $\bropt(G,s)$ to refer to the number of rounds in the optimal broadcast scheme. 

\begin{definition}
    \emph{Cactus graphs} are 
    connected graphs in which any two simple cycles 
    have at most one vertex in common.
\end{definition}

\begin{definition}
\label{def:ourflowergraph}
A tree $T$ is said to be a \emph{\cater} if there are three special nodes $x, y,$ and $z$ in $T$ such that every node in $T$ is either located on the path between $x$ and $y$ or is connected to $z$. 

A graph $G$ is said to be a \flower, if and only if it has a \pistil vertex $c$ such that $G\setminus \{c\}$ is a set of disjoint reduced caterpillars 
such that $c$ is
connected to exactly two vertices in any of these caterpillars, none being special vertices.
\end{definition}

An example of \flower graphs and one of its corresponding reduced caterpillar components is shown in Figure~\ref{fig:flowergraphs_example}. Informally, a \flower graph is formed by a set of cycles that have a common \pistil $c$; moreover, each cycle has a special vertex ($z$ vertices). Any vertex in $G$ is either 
i) a part of one of the cycles or 
ii) a part of a ``dangling path'' connected to a neighbor of $c$ or 
iii) finds a special vertex as its sole neighbor.

\begin{figure}
    \centering
    \begin{minipage}[b]{0.42\textwidth}
        \centering
        \includegraphics[width=0.38\linewidth]{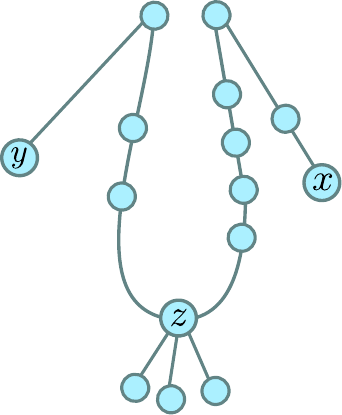} % Image 2
        \subcaption{An example of a reduced caterpillar, with the special vertices $x$, $y$, and $z$ highlighted} % Caption for Image 2
        \label{fig:flowerGraph_caterpillar}
    \end{minipage}
    \hfill
    \begin{minipage}[b]{0.52\textwidth}
        \centering
        \includegraphics[width=0.65\linewidth]{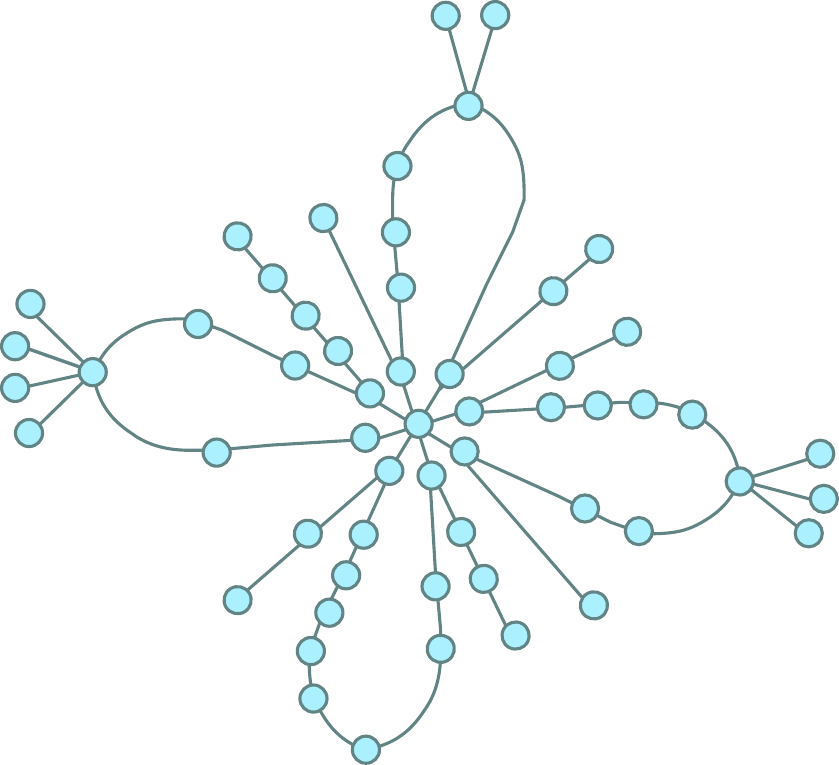} % Image 1
        \subcaption{An example of a \flower graph} % Caption for Image 1
        \label{fig:flowerGraph_cdsToB}
    \end{minipage}
    \caption{An illustration of reduced caterpillar and \flower graphs}
    \label{fig:flowergraphs_example}
\end{figure}

\newcommand{\tsub}{\textsubscript}

\begin{definition}[\cite{robertson1983pathwidth}]
A path decomposition $D$ of a given graph $G = (V,E)$ 
    is a sequence $\langle B_1, B_2, \dots, B_k\rangle$, where each $B_i$ is called a bag and contains a subset of $V$, such that every vertex $v \in V$ appears in at least one bag and, for every edge $(u, v) \in E$, there exists a bag $B_i$ containing both $u$ and $v$. Furthermore, 
    if a vertex $v$ appears in $B_i$ and in $B_j$, it must appear in any $B_k$ where $k\in[i,j]$.

    The width of the path decomposition $D$ is the maximum cardinality of its bags minus $1$. Now, $G$ is said to have pathwidth $w$ if it has a path decomposition of width at most $w$. 
\end{definition}

\begin{observation}
\label{obs:flower-pathwidth}
    \Flower graphs have a pathwidth of at most 2. 
\end{observation}

\begin{proof}
    Removing the \pistil from a \flower graph results in a collection of disjoint caterpillar graphs, each having a pathwidth of $1$ (caterpillars are graphs of pathwidth $1$~\cite{ProskurowskiT99}). 
    A valid path decomposition for a \flower graph can be achieved by adding the \pistil center to all bags of path decompositions of these caterpillars. 
\end{proof}

\section{\ouralgo: A 2-Approximation for Cactus Graphs}
\label{sec:2approx}

In this section, we present our 2-approximation algorithm for cactus graphs. We use ideas from \fast model with parameter $k$~\cite{grigni1991tightkbroad}, where a vertex can inform up to $k$ of its neighbors in a single round via a 
\emph{super call}. It is easy to see that if one can complete $k$-broadcasting in $m$ rounds, then it is possible to complete broadcasting (in the classic setting) within $km$ rounds. This can be achieved by ``simulating'' a super-call with up to $k$ regular calls (see  Lemma~\ref{lemma:2approx2calls}).
For our algorithm, we use \fast with $k=2$. In particular, we design an algorithm for \twofast in a cactus graph $G$ and show that if it completes within $m$ rounds, then any broadcast scheme for classic broadcasting in $G$ takes at least $m$ rounds (see Lemma~\ref{lemma:2approxOPT}). Therefore, if we simulate every super call with two regular calls (in arbitrary order) using Lemma~\ref{lemma:2approx2calls}, the broadcasting completes within $2m$ rounds, and thus, we achieve an approximation factor of $2$. 
% In what follows, we describe the algorithm in detail. 

\subsection{\Fast Model}

In the \fast model, an informed vertex can simultaneously inform up to $k$ neighbors in a single round, a process we refer to as a \textbf{super call}.  
Creating networks that allow fast broadcasting under this model has been studied in previous work~\cite{harutyunyan2001improved}. 
We will present a method to convert a broadcasting schema in the \fast model into the classic model (without super calls). The final number of rounds in the classic model will be at most $k$ times the number of rounds in the \fast model. This model applies not only to cactus graphs but also to every arbitrary graph.

\begin{lemma} \label{lemma:2approx2calls}
Let $S_k$ be a 
broadcast schema for graph $G$ in the \fast model. It is possible to convert $S_k$ to a broadcast scheme $S$ in linear time for graph $G$ in the classic model such that broadcasting in $S$ completes within $k$ times the number of rounds as broadcasting in $S_k$, i.e., $\br(S) \leq k\cdot \br(S_k)$. 
\end{lemma}
\begin{proof}
Form $S$ from $S_k$ as follows. Suppose a vertex $v$ informs its neighbors at rounds $(1,2,\ldots, p)$ in $S_k$. In the $S$, $v$ informs the same neighbors that it informs in $S_k$ (i.e., they both have the same broadcast tree), except that each super-call in $S_k$ is replaced by up to $k$ regular calls in $S$, ensuring that if a neighbor $x$ is informed before a neighbor $y$ in $S_k$, then $x$ gets informed before $y$ in $S$ as well. For example, if a vertex $v$ first informs a neighbor $a$ with a regular call and then neighbors $b,c$ simultaneously with a super call right after getting informed in $S_k$, it will inform $a$ at round 1 and $b,c$ at rounds 2 and 3 (in arbitrary order) in $S$ (see Figure~\ref{fig:2approxFastRegular} for an illustration).

\begin{figure}
    \centering
    \begin{minipage}[b]{0.45\textwidth}
        \centering
        \includegraphics[width=0.55\linewidth]{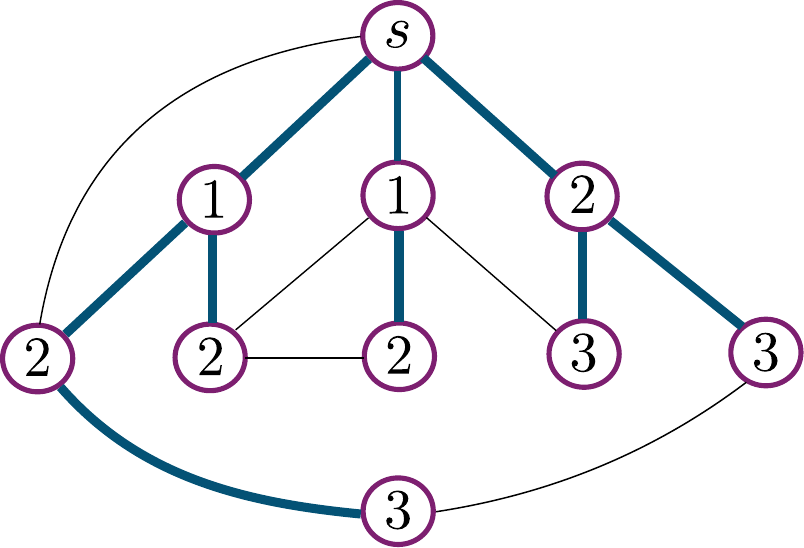} % Image 1
        \subcaption{The broadcast scheme in the \fast model with $k=2$
        } % Caption for Image 1
        \label{fig:2approxFastmodel}
    \end{minipage}
    \hfill
    \begin{minipage}[b]{0.45\textwidth} 
        \centering
        \includegraphics[width=0.55\linewidth]{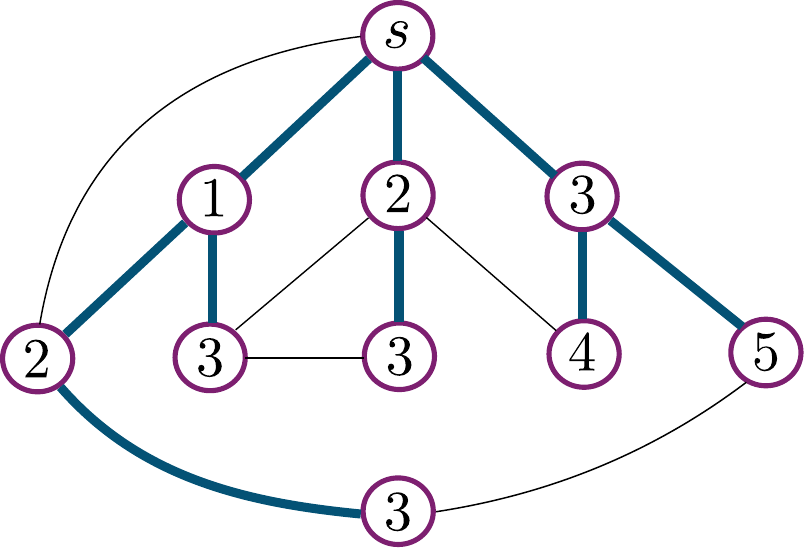} % Image 2
        \subcaption{The corresponding classic broadcasting scheme} % Caption for Image 2
        \label{fig:2approxRegular}
    \end{minipage}
    \caption{An illustration of Lemma~\ref{lemma:2approx2calls}. 
    Highlighted edges show broadcast trees, 
    and the numbers indicate the rounds at which vertices get informed.}
    \label{fig:2approxFastRegular}
\end{figure}

Next, we prove that broadcasting in $S$ completes within a factor $k$ of $S_k$. 
We let $t_{S_k}(x)$ (respectively $t_{S}(x)$) denote the round at which vertex $x$ gets informed in $S_k$ (respectively in $S$). 
We will use an inductive argument on the time that vertices get informed to show for any vertex $x$, we have $T_{S}(x) \leq k T_{S_k}(x)$. 
At $t=0$, we have 
$T_{S}(s) =  T_{S_k}(s) = 0$, and the base of the induction holds. Suppose a vertex $v$ is informed at time $t_{S_k}(v)$ in $S_k$, and 
assume $v$ informs its neighbor $x$ exactly $d$ rounds after it is informed. That is, $x$ gets informed in $S_k$ at time $T_{S_k}(x) = T_{S_k}(v) + d$. In $S$, $v$ informs $x$ at most $kd$ around after getting informed (because each super call is replaced by at most $k$ regular calls). 
Therefore, $T_{S}(x) \leq T_{S}(v) + kd$. On the other hand, by the induction hypothesis, we have $T_{S}(v) \leq k T_{S_k}(v)$. We can conclude that $T_{S}(x) \leq k T_{S_k}(v) + kd = k T_{S_k}(x)$, which completes the induction step.
\end{proof}

\subsection{\ouralgo Algorithm}
This section presents \ouralgo, a \twofast algorithm 
that yields our $2$-approximation algorithm for the \telebr problem in cactus graphs. 
\ouralgo is a mutually recursive algorithm with two main methods, namely, 
    \singlefunc{$G, s$} and
    \doublefunc{$G, s_1, s_2$}.
The \singlefunc{$G, s$} method is used for \twofast a connected subgraph $G$ starting from a source vertex $s$. The \doublefunc{$G, s_1, s_2$} method is used for \twofast a connected subgraph $G$ using two sources $s_1$ and $s_2$ under an assumption that both $s_1$ and $s_2$ have the message at time $0$ and there is a unique path between them in $G$.

\subparagraph*{Broadcasting from a single source.}
We explain how \singlefunc{$G, s$} operates. Observe that removing the source vertex $s$ partitions $G$ into one or more disjoint connected components. By the definition of cactus graphs, each of these connected components contains at most two neighbors of $s$ 
(see Observation~\ref{observation:two-neighbors}). We refer to a connected component with one neighbor (respectively two neighbors) of $s$ as a \emph{single-neighbor} (respectively \emph{double-neighbor}) component (see Figure~\ref{fig:2approxSingleSource}).

\figcomponent

\begin{observation}
\label{observation:two-neighbors}
Let $H$ be a connected induced subgraph of a cactus graph $G$. Any vertex $v \in G$ that is not in $H$ finds at most 2 neighbors in $H$.
\end{observation}
\begin{proof}
    Since $H$ is connected, having more than two neighbors for $v$ would imply an edge belonging to at least two cycles, which contradicts the definition of cactus graphs. 
\end{proof}

\singlefunc{$G,s$} computes the broadcast time of each single-neighbor component $C$ 
recursively by computing \singlefunc{$C, u$}, where $u$ is the single neighbor of $s$ in $C$. Similarly, it computes the broadcast time of each double-neighbor component $C$ by computing 
\doublefunc{$C, u, v$}, where $u, v$ are neighbors of $s$ in $C$. 
After calculating all these broadcasting times, \singlefunc{$G,s$} sorts these values in non-increasing order and informs the neighbors of $s$ in this order. For that, it uses a regular call for single-neighbor components and a super-call for double-neighbor components. 
In other words, \singlefunc{$G,s$},  informs 
the component with the largest broadcasting time in the first round, the component with the second largest broadcast time next, and so on. This is reminiscent of broadcasting in trees~\cite{fraigniaud2002polynomial}, except that super-calls are used for double components. Let $b_i$ denote the broadcast time of the $i$'th component $C_i$ in this order. 

The total broadcasting time for $\singlefunc{G,s}$ is then computed as $\max_i (i+b_i)$, which will be the output of $\singlefunc{G,s}$.

\subparagraph*{Broadcasting from two sources.} Next, we describe \doublefunc{$G, s_1, s_2$}. 
Let $P$ be the unique path between $s_1$ and $s_2$.
In any broadcast tree $T$ of
\doublefuncc, some vertices on $P$ are informed through $s_1$ and some through $s_2$. Thus, 
exactly one edge $e \in P$ is excluded from the tree. 
After removing such $e$, the graph $G$ will be partitioned into two disjoint connected components $G_1^{e}$ and $G_2^{e}$, which are respectively informed through $s_1$ and $s_2$ (See Figure~\ref{fig:2approxDoubleMerge}). 
The method \doublefuncc works by removing the edge $e^*$
with a minimum value of 
\linebreak $\max \{ \singlefuncc (G_1^{e},s_1), \singlefuncc (G_2^{e},s_2) \}$ 
over all $e\in P$. 

Next, we explain how \doublefuncc finds $e^*$. An exhaustive approach that 
tries all edges $e\in P$ and calls \singlefuncc twice per edge may take exponential time. 

To fix this, 
\doublefuncc works in two phases:
\begin{itemize}
    \item In {Phase} ${1}$, the method pre-computes the broadcast time for the connected subgraphs of $G$ after removing vertices of $P$ from $G$. By Observation~\ref{observation:two-neighbors}, 
every resulting connected component after removing $P$ has at most two neighbors in $P$. For any single-neighbor component $C$ with vertex $s$ connected to $P$, \singlefunc{$C,s$} is computed recursively. Similarly, for any double-neighbor component with vertices $s_1, s_2$ connected to $P$,  \doublefunc{$C,s_1,s_2$} is computed recursively.

\item Next, we describe {Phase} ${2}$. At each iteration of Phase $2$, we fix an edge $e=(u,v) \in P$ and aim to compute \singlefunc{$G_1^e,s_1$} and \singlefunc{$G_2^e,s_2$}, where $u\in G_1^e$ and $v\in G_2^e$. 
We explain how to efficiently compute \singlefunc{$G_1^e,s_1$}. Finding \singlefunc{$G_2^e,s_2$} is done similarly. 
Let $\langle u(=u_1), u_2, \ldots, u_k(=s_1)\rangle$ be the sequence of vertices in the unique path from $u$ to $s_1$.
Let $H^e_{i}$ be the connected component of $G$ that contains $u_i$ after removing the edge ($u_i, u_{i+1}$) from $G_1^e$. 
Note that $H^e_k = G_1^e$. 
We compute $b_i=$\singlefunc{$H^e_i, u_i$} using the values computed in Phase $1$ as follows in an iterative manner from $i=1$ to $i=k$. 

Consider all connected components of $H^e_i$ after removing $u_i$. Note that, for $i>1$, $H^e_{i-1}$ is one of these components. 
Consider the set $B$ of broadcast times for all these components, 
with the neighbor(s) of $u_i$ as the source(s) of broadcast. All these values are computed in Phase $1$ except for $b_{i-1}$, which is computed in the previous iteration. \doublefuncc sorts $B$ in non-increasing order of broadcast times and informs neighbors of $u_i$ accordingly (similar to \singlefuncc). As a result, in the $i$'th iteration, the broadcast time $b_i$ is computed as $\max_{j \in [B]} (j+B[j]) $. At iteration $i=k$, 
we compute $b_k$, which is indeed \singlefunc{$G^e_1, s_1$} (see Observation~\ref{obs:doublefunc}).
\end{itemize}
Figure~\ref{fig:2approxDoubleMerge} provides an illustration.

% \twoapproxfigcomponent
\twoapproxnewfig

\begin{restatable}{theorem}{twoapprox} \label{thm:2approx}
    \ouralgo runs a polynomial-time for \telebr and achieves an approximation factor of $2$ on cactus graphs.
\end{restatable}

%\begin{proof}[Sketch]
\paragraph{Proof Overview.}
We provide an overview of the proof before a formal argument.
First, we show that \ouralgo runs in $\oh(n^3 \log n)$. This can be established using induction on the size of the input graph, based on the recursive nature of the algorithm. Details can be found in Lemma~\ref{lemma:2approxpoly}.

Second, we show that \ouralgo has an approximation factor of $2$. Intuitively, \singlefuncc with super-calls runs no longer than an optimal broadcast scheme without super-calls; this is because every super call only expedites informing one of the two sources of broadcasting in double-neighbor components. As mentioned earlier (Lemma~\ref{lemma:2approx2calls}), a broadcast scheme with super calls can be simulated with a scheme (without super-calls) that completes no later than twice the number of rounds. 
On the other hand, the selection of edge $e^*$ by \doublefuncc ensures that 
$
    \max (\singlefunc{$G^{e^*}_1, s_1$}, \singlefunc{$G^{e^*}_2, s_2$}) \leq \max (\singlefunc{$G^{e^+}_1, s_1$}, \singlefunc{$G^{e^+}_2, s_2$})
$, 
where $e^+$ is the edge that is absent in the optimal broadcast tree of $G$. 
A formal proof follows from an inductive argument on the input size (see Lemma~\ref{lemma:2approxOPT}).

\begin{observation}\label{obs:doublefunc}
Let $\langle u_1 (=u), u_2,\ldots, u_k (=s_1) \rangle$ be the vertices on the path from $u_1$ (an endpoint of the removed edge $e$) and the source $s_1$. Then, for any $i\in [k]$, the broadcast time $b_i$ for $(H^e_i,u_i)$ in \doublefuncc equals to \singlefuncc applied to $H_i$, that is, $b_i = $ \singlefunc{$H^e_i, u_i$}.
\end{observation}
\begin{proof}
    We use induction on $i$. For $i=1$, \doublefuncc simply uses the precomputed broadcast times $B_i$ of the connected components after removing the source $u_1$ from $H_1^e$ in Phase 1. It proceeds with sorting these broadcast times and setting $b_i = \max_{j\in[|B_i|]} (j+B_i[j])$. This is what \singlefuncc does by definition when broadcasting from $u_1$ in $H_1$. The inductive step is similar, except that when computing $b_i$, the value of $b_{i-1}$ is also in $B_i$. This value is indeed \singlefunc{$H^e_{i-1},u_{i-1}$}, by the induction hypothesis. 
\end{proof}

\begin{lemma} \label{lemma:2approxpoly}
\ouralgo terminates in $\oh(n^3 \log n)$ time.
\end{lemma}
\begin{proof}
To prove the lemma,
we define the following functions. The functions $T_s(n)$ and $T_d(n)$ represent the maximum time required to process a subgraph of size $n$ using \singlefuncc and \doublefuncc, respectively. Let $T(n)=\max (T_s(n), T_d(n))$.

\singlefunc{$G,s$} involves running the recursive methods with time $T(x)$ on any component of size $x$ after removing $s$,  
along with an additional $\oh(n\log n)$ time to sort the 
broadcasting times obtained from recursive calls. On the other hand, \doublefuncc includes precomputing the broadcast time for the subgraphs after deleting the vertices of path $P$ between $s_1$ and $s_2$. For each of $\oh(n)$ edges $e\in P$ that it considers for removing, we have $k$ iterations, each computing \singlefunc{$H_i,u_i$} (for $i\in [k]$). 
At each iteration, we sort the broadcast times for connected components of $H_i$ after removing $u_i$, which takes $\oh(d(u_i) \log d(u_i))$, where $d(u_i)$ is the degree of $u_i$ in $G^e_1$. Therefore, for a fixed edge, the total time would be $\sum_{i\in[k]} \oh(d(u_i) \log d(u_i)) = \oh(n \log n)$. Given that there are $\oh(n)$ possible edges to remove, the total time complexity of \doublefuncc would be $\oh(n^2 \log n)$ plus the additional time spent in Phase $1$.  

For each $n$, let $\mathcal{Q}_n$ be the
set of partitions of any integer less than $n$ to smaller integers. That is, any set of integers that sum to at most $n-1$ is a member of $\mathcal{Q}_n$.
To summarize, there are some $n_0,c' \geq 0$ such that for large $n \geq n_0$, we have:

\scalebox{.9}{
\begin{minipage}{1.2\textwidth}
\begin{align*}
    T_s(n) &\leq c'(n\log n) +\max_{S \in \mathcal{Q}_n} \sum_{x \in S} T(x), \\
    T_d(n) &\leq c'(n^2\log n) + \max_{S \in \mathcal{Q}_n} \sum_{x \in S} T(x).
\end{align*}
\vspace{1mm}
\end{minipage}}
Therefore,

\scalebox{.9}{
\begin{minipage}{1.05\textwidth}
\begin{align}
    T(n)=&\max (T_s(n), T_d(n)) \nonumber\\
    \leq& c'n^2\log n+ \max_{S \in \mathcal{Q}_n} \sum_{x \in S} T(x).\label{inequality:cactusalgo:time}
\end{align}
\vspace{1mm}
\end{minipage}}

We claim that $T(n) \leq cn^3 \log n$, in which $c=\max(c', T(n_0))$. To prove the claim, we use induction on $n$. For the base case of $n=n_0$, as $c\geq T(n_0)$, the inequality clearly holds.
Next, assume $T(n') \leq cn'^3 \log (n')$ for all $ n' \in [n_0,n-1]$. 

\scalebox{.9}{
\begin{minipage}{1.2\textwidth}
\begin{align*}
     T(n) &\leq  c'n^2\log n + \max_{S \in \mathcal{Q}_n} \sum_{x \in S} T(x) && \text{based on \eqref{inequality:cactusalgo:time}}\\
     & \leq c'n^2\log n  +\max_{S\in \mathcal{Q}_n} \sum_{x \in S} cx^3 \log (x) && \text{using induction hypothesis}\\ 
     & \leq c'n^2\log n + c (n-1)^3 \log n  && \text{$\forall_{S\in \mathcal{Q}_n} \sum_{x \in S} x<n$} \\ 
     & < c'n^2\log n +(cn^3 - cn^2) \log n  \\
     & \leq cn^3 \log n.  && \text{$c \geq c'$}
\end{align*}
\vspace{1mm}
\end{minipage}}
which completes the induction step, and we can conclude $T(n)\leq c n^3\log n$.
\end{proof}

\begin{lemma} \label{lemma:2approxOPT}
Consider an instance  $(G,s)$ of the broadcasting problem in a cactus graph $G$. Let $\br(G, s)$ denote the number of rounds that it takes for \ouralgo to complete \fast in $G$ with $k=2$, and $\bropt(G, s)$ denote the optimal number of rounds for broadcasting in the classic model. Then we have $\br(G, s)\leq \bropt(G, s)$.
\end{lemma}
\begin{proof}
    We use induction on the number of vertices in $G$ to prove the claim for both 
    methods \singlefuncc and \doublefuncc. In the base cases, the graph consists of just one (for \singlefuncc) or two (for \doublefuncc) source vertices, and the statement trivially holds. 
    
    For the induction step, we prove the claim for a graph with $n$ vertices, assuming it holds for all graphs with $n' < n$ vertices.

    First, we consider the simpler case of \doublefunc{$G, s_1, s_2$}. Recall that the algorithm 
    finds an edge $e^* \in P$ 
    that minimizes 

    \scalebox{.9}{
    \begin{minipage}{1.2\textwidth}
    \begin{align*}
        \max (\singlefunc{$G^{e^*}_1, s_1$}, \singlefunc{$G^{e^*}_2, s_2$}).
    \end{align*}
    \vspace{1mm}
    \end{minipage}}
    Assume that the optimal solution removes an edge $e^+ \in P$. Based on the induction hypothesis we know that $\bropt(G^{e^+}_1, s_1) \geq \singlefunc{$G^{e^+}_1, s_1$}$ and $\bropt(G^{e^+}_2, s_2) \geq \singlefunc{$G^{e^+}_2, s_2$}$. In this case, the broadcast time of the optimal solution would be

    \scalebox{.9}{
\begin{minipage}{1.2\textwidth}
    \begin{align*}
        \bropt(G,s) &= \max(\bropt(G^{e^+}_1, s_1), \bropt(G^{e^+}_2, s_2)) \\
        &\geq \max(\singlefunc{$G^{e^+}_1, s_1$}, \singlefunc{$G^{e^+}_2, s_2$}) && \text{induction hypothesis}\\
        &\geq \max(\singlefunc{$G^{e^*}_1, s_1$}, \singlefunc{$G^{e^*}_2, s_2$}) && e^*\text{ definition}\\
        &=\doublefunc{$G, s_1, s_2$}.
    \end{align*}
    \vspace{1mm}
\end{minipage}}

    Next, we address the case of the \singlefunc{$G, s$}. 
    Consider the optimal broadcast scheme $T^*$ in the classic model. For a component $C_i$ with up to two neighbors of $s$, let $\pi^1_i$ and $\pi^2_i$ be the rounds at which $s$ informs its first and its second (if exists) neighbor in $C_i$ in $T^*$, respectively. Note that $\pi^1_i < \pi^2_i$, and we let $\pi^2_i=\infty$ if $s$ informs just one vertex in $C_i$. If both neighbors of $s$ in $C_i$ were informed simultaneously via a super call in round $\pi^1_i$, the broadcast time would not increase compared to informing the first one at $\pi^1_i$ and the second one at round $\pi^2_i$. Hence, the following inequality holds.

    \scalebox{.9}{
    \begin{minipage}{1.2\textwidth}
    \begin{align*}
        \bropt(G,s) &= \max_i (\bropt(C_i, (\pi^1_i, \pi^2_i)))\\
        & \geq \max_i (\bropt(C_i, (\pi^1_i, \pi^1_i)))\\ 
        & = \max_i (\bropt(C_i, (0, 0))) + \pi^1_i, \nonumber 
    \end{align*}
    \vspace{1mm}
    \end{minipage}}
    where $\bropt(C_i, (\pi^1_i, \pi^2_i))$ is the optimal broadcasting time in the classic model given that the first source is informed at round $\pi^1_i$ and the second one (if exists) at round $\pi^2_i$.
    Recall $b_i$ is the broadcast time for each component $C_i$, which is determined by applying the recursive methods on $C_i$. This is equivalent to informing the neighbors of $s$ in $C_i$ at time $0$.

    \scalebox{.9}{
    \begin{minipage}{1.2\textwidth}
    \begin{align*}
        \bropt(G,s) &\geq \max_i (\bropt(C_i, (0, 0))) + \pi^1_i \\
        &\geq \max_i (b_i + \pi^1_i).  && \bropt(C_i, (0, 0)) \geq b_i \text{ (induction hypothesis)}
    \end{align*}
    \vspace{1mm}
\end{minipage}}
    
    \singlefuncc selects $\pi'_i$ such that $\max_i (b_i + \pi'_i)$ is minimized. This is because it sorts the broadcast times required to inform each subgraph in non-increasing order and informs the neighbor(s) of $s$ in $C_i$ accordingly. We can write
    
    \scalebox{.9}{
    \begin{minipage}{1.2\textwidth}
    \begin{align*}
        \bropt(G,s) &\geq \max_i (b_i + \pi^1_i) \\
        &\geq \min_{\pi'} (\max_i (b_i + \pi'_i))\\
        &= \br(G,s).
    \end{align*}
    \vspace{1mm}
\end{minipage}}
    Thus, $\bropt(G,s) \geq \br(G,s)$. Thus, the number of super rounds in both methods does not exceed the number of rounds used by the optimal broadcast scheme in the classic model.
\end{proof}

From the above results, we can establish the proof of Theorem~\ref{thm:2approx} as follows.

\begin{proof}[of Theorem~\ref{thm:2approx}]
\label{proof:2approx}
    By Lemma~\ref{lemma:2approxpoly}, we showed that \ouralgo terminates in polynomial time. Furthermore, Lemma~\ref{lemma:2approxOPT} establishes that the \twofast scheme of \ouralgo completes no later than the optimal broadcast time. Moreover, by Lemma~\ref{lemma:2approx2calls}, this \twofast scheme can be converted to a scheme in classic broadcasting that completes within twice the optimal broadcast time.   
\end{proof}

The approximation factor given by Theorem~\ref{thm:2approx} is tight. Consider a graph $G$ formed by $t$ triangles that all share an endpoint $s$. The broadcasting time of \ouralgo on $(G,s)$ is $2t$ as its broadcast tree will be a star with $2t$ leaves, while the optimal broadcast tree completes broadcasting in $t+1$ rounds ($s$ will have degree $t$). 
Thus, the approximation factor tends to $2$ for large $t$.

\section{Hardness Proof of \Flower Graphs}
\label{sec:hardness}

The NP-hardness proof of \telebr in \flower graphs proceeds through successive reductions. First, we reduce \sssattt, which is known to be NP-hard~\cite{tovey1984simplified} to a new problem that we call \tislong (Lemma~\ref{thm:1}), which itself reduces to another new problem \dsprlong (Lemma~\ref{thm:TIS2DOME}). Finally, we present a reduction from \dspr to \telebg, which establishes our main result (Theorem~\ref{thm:hardness-flower}).

\subsection{Hardness of \tislong}\label{sec:tishard}
The first step in reducing \sssattt to \telebg is to establish a reduction from \sssattt to \tislong (\tis). For that, we first explain how 
to construct a \tis instance from a given \sssattt instance and demonstrate the equivalence of solutions between the two problems. 

\begin{definition}[\cite{tovey1984simplified}]

The \emph{\sssattt} problem is a special type of the classic $3$-satsifiabiity (\sattt) problem. The input is a boolean CNF formula $\phi=(X, C)$, where $X$ is the set of variables and $C$ is the set of clauses 
such that every clause $C_i \in C$ contains \emph{exactly} three literals $\ell^i_k\in C_i$ for $k\in [3]$, and each variable $x\in X$ appears in at most four clauses. The decision problem asks whether there is a satisfying assignment of $\phi$.
\end{definition}

The \sssattt problem is known to be NP-complete \cite{tovey1984simplified}. 
In what follows, we call a pair of intervals ($I_i$, $\overline{I_i}$) for $i\in [n]$ a \emph{twin interval}, assuming that $I_i$ and $\overline{I_i}$ are non-crossing and have equal lengths with endpoints in $[m]$ for some positive $m$. The endpoints of $I_i$ are less than the endpoints of $\overline{I_i}$. We refer to $I_i$ and $\overline{I}_i$ the \emph{left} and the \emph{right} interval of the twin, respectively.

\begin{definition}

An instance of the \emph{\tislong (\tis)} is formed by a tuple $(I,r,m)$, where $I$ is a set of \emph{twin} interval pairs and $r$ is a \emph{restriction function} with domain $[m]$, where $m$ is referred to the \emph{horizon} of the instance. We have $I = \{(I_1, \overline{I_1}), \ldots, (I_n, \overline{I_n})\}$, where  for each $i\in [m]$, $(I_i, \Bar{I_i})$ are non-crossing intervals with endpoints in $[m]$.

The objective is to select exactly one of $I_i$ and $\overline{I_i}$, while respecting the restriction imposed by the restriction function as follows. The restriction function $r : [m] \to [n]$ requires that for any $t \in [m]$, the number of selected intervals that contain $t$ be at most $r(t)$. The decision problem asks whether there is a valid selection that satisfies these restrictions. 
\end{definition}
For a solution $S$ and $t\in[m]$, let $\Gamma_S(t)$ be the number of selected intervals in $S$ containing $t$.

% An instance of \tis is shown in Figure~\ref{fig:TIS-example}. 

% \dsprfigcomponent

\subparagraph*{Construction.} 
Suppose we are given an instance $\phi=(X, C)$ of \sssattt. We construct an instance $\mathcal{I}_\phi$ of \tis as follows. 
For each variable $x_i \in X$, 
form a pair of $(X_i,\overline{X_i})$ of twin intervals, each of length $7$, where $X_i = [16i - 15, 16i - 8]$ and 
$\overline{X}_i = [16i - 7, 16i]$. 
We refer to $X_i$ (respectively $\overline{X_i}$) as the \emph{principal} interval of literal $x_i$ (respectively $\neg x_i$). 
The length $7$ of these intervals allows for placing up to $4$ non-crossing unit intervals, each of length $1$, within the principal interval; these intervals will be used for clause gadgets, as we will explain.
For each clause $C_j$, we define three interval twin pairs $(I_k^j,\overline{I_k^j})$, each associated with one of the literals of $\ell^j_k\in C_j$ for each $k\in[3]$. Suppose $\ell^j_k$ is the $p$'th literal where its corresponding variable $x_i\in X$ appears ($p\in[4]$). If $\ell^j_k$ is a positive (respectively negative) literal of variable $x_i$, 
 then define $I_k^j$ as a unit interval starting at $16i-7+2(p-1)$ 
 (respectively, starting at $16i-15+2(p-1)$).
Intuitively, $I_k^j$ is defined as one of the four-unit intervals within the principal interval of $\neg x$.
Meanwhile, $\overline{I_k^j}$ is $[16n+2j, 16n+2j+1]$ for all $k \in [3]$. 
Finally, to define the restriction function, we let $r(t)=1$ for all $t\leq 16n$ and $r(t)=2$, otherwise. Figure~\ref{fig:SAT2TIS-example} illustrates an example of the above reduction. 

\begin{figure}
	\centering
	\includegraphics[width=\linewidth]{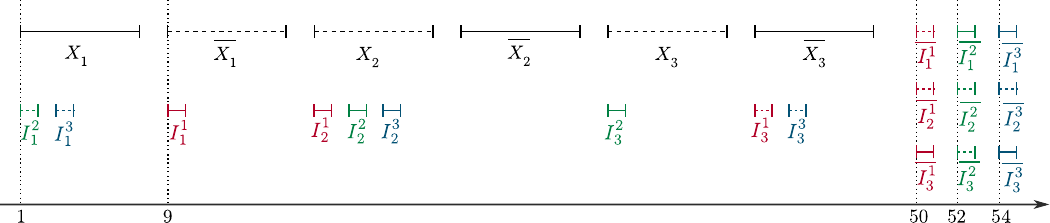}
	\caption{Construction of the instance $\mathcal{I}_\phi$ of the \tis problem 
    from the \sssattt instance $\phi = (x_1 \lor \neg{x_2} \lor x_3) \land (\neg{x_1} \lor \neg{x_2} \lor \neg{x_3}) \land (\neg{x_1} \lor \neg{x_2} \lor x_3) $. The restriction function $r(t)$ is defined as $r(t) = 1$ for $t \leq 48$ and $r(t) = 2$ otherwise. A satisfying assignment for $\phi$ is $\sigma(x_1)=1$ and $\sigma(x_2) = \sigma(x_3)=0$; the selected intervals in the corresponding solution for $\mathcal{I}$ are shown in solid color.}
    \label{fig:SAT2TIS-example}
\end{figure}

\subparagraph*{Correctness.}
First, we provide an intuitive explanation of why this reduction works. We argue that a satisfying assignment to instance $\phi$ of the 3,4-SAT problem bijects to a satisfying interval selection for the instance $\mathcal{I}_\phi$ of \tis. The bijection requires that for any literal $\ell$ that is true in a satisfying assignment of $\phi$, the principal interval of $\ell$ is selected in the solution for $\mathcal{I}_\phi$. 

The restrictions imposed by $r$ imply that whenever a principal interval of a literal $\ell$ is selected, the unit intervals associated with (up to $4$) occurrences of $\neg \ell$ 
cannot be selected (because $r(t) = 1$ in their intersecting points). Thus, whenever the principal interval of a literal $\ell$ is selected (i.e., its twin, which is the principal interval of $\neg \ell$ is not selected), one can select the unit intervals associated with any occurrence of $\ell$ in clauses where it appears; this is because they intersect with the principal interval of $\neg \ell$, which is not selected. Selecting these unit intervals means their twin unit intervals
This yields an equivalent between a satisfying assignment in the \sat formula and the \tis. 
We further note that the number of twin intervals in $\mathcal{I}_\phi$ is polynomial in $|X|$ and $|C|$. We can conclude the following lemma, which establishes the hardness of \tis in the strong sense based on the above reduction.

To establish the main result of this section,
%To prove Lemma~\ref{thm:1}, 
we first establish the two sides of the reduction in the following lemmas. 

%*************************
\begin{lemma} \label{lemma2_SATtoTIS}
If the answer to the instance $\mathcal{I}_\phi$ of \tis is yes, then the instance $\phi$ of \sssattt is satisfiable.
\end{lemma}

\begin{proof}
Let $S$ be a selection of intervals for $\mathcal{I}_\phi$ that certifies a yes answer to $\mathcal{I}_\phi$. We construct an assignment $\sigma: X \to \{0, 1\}$ for $\phi$ as follows: For each variable $x_i \in X$, set $\sigma(x_i) = 1$ if the principal interval $X_i$ is in $S$, and set $\sigma(x_i) = 0$, otherwise.
We will prove that $\sigma$ satisfies $\phi$. By construction, for each clause $C_j$ we know that all three unit intervals $\overline{ I^j_k}$ (for $k\in[3]$) contain point $t=16n+2j$, where $r(t)=2$. Thus, at least one of these three unit intervals is not included in $S$. Let $\overline{I^j_q}$ for some $q\in [3]$ be such interval that is not selected; that is, $I^j_q$ is in $S$. Because $I^j_q$ overlaps with the principal interval of $\neg {\ell^j_q}$ at some point $t \leq 16n$, and $r(t)=1$, the principal interval of $\neg {\ell^j_q}$ cannot be in $S$. Namely, $S$ contains the principal interval of $\ell^j_q$. Therefore, in the corresponding assignment of $\phi$, we have $\sigma(\ell^j_q)=1$, and thus the clause $C_j$ is satisfied. We conclude that each clause $C_j$ in $\phi$ has at least one true literal with true value $\sigma$, and thus $\sigma$ satisfies $\phi$.
\end{proof}

%*************************
\begin{lemma} \label{lemma1_SATtoTIS}
If an instance $\phi=(X,C)$ of the \sssattt problem is satisfiable, then the answer to the instance $\mathcal{I}_\phi$ of \tis is yes.
\end{lemma}

\begin{proof}
Let $\sigma: X \to \{0, 1\}$ be a satisfying assignment for $\phi$. We construct a solution $S$ for $\mathcal{I}_\phi$, which contains the selected interval from each twin pair. For each variable $x \in X$, if $\sigma(x) = 1$, we select the principal interval of $x$ along with the left (respectively right) intervals of all the unit intervals associated with positive (respectively negative) literals of $x$. 
Similarly, if $\sigma(x) = 0$, we select the principal interval of $\neg x$, along with the left (respectively right) intervals of all the unit intervals associated with negative (respectively positive) literals of $x$.

Next, we prove that $S$ satisfies all restrictions in $\mathcal{I}_\phi$. 
First, we consider $t\leq 16n$, and show that the restriction imposed by $r(t)\leq 1$ is satisfied. In this range of $t$, at most two intervals in $\mathcal{I}_t$ contain $t$. 
We need to show whenever $S$ includes the principal interval of a variable $x$, it does not contain the unit intervals that intersect the principal interval of $x$. 
This holds because these unit intervals are associated with $\neg x$, which is not selected by the construction of $S$.
Consequently, the number of selected intervals at the intersection points of the principal interval of $x$ and the unit intervals is at most $1$, satisfying the restriction $r(t) = 1$ (for $t \leq 16n$).

Next, we consider $t>16n$ and show the restriction imposed by $r(t) = 2$ is satisfied. There is one unit interval for each literal in a clause $C_j$, i.e., for any $t>16n$, there are up to 3 intersecting intervals. Since $\sigma$ satisfies $\phi$, at least one literal in each clause receives a true value. 
Suppose that it is the $q$'th literal for $q\in [3]$. 
As a result, based on the construction of $S$, $\overline{I^j_q}$ is not selected. This ensures that the number of selected intervals at these points is at most two, satisfying the restriction $r(t) =2$. 

We conclude that the restrictions imposed by $r(t)$ are satisfied for all $t\in[m]$ and thus $S$ is a valid solution for $\mathcal{I}_\phi$.
\end{proof}

%We are now ready to prove Lemma~\ref{thm:1}.

\begin{restatable}{lemma}{tishard}\label{thm:1}
Answering instances $(I,r,m)$ of the \tis problem is NP-hard even if the $m$ is polynomial in $|I|$. 
\end{restatable}

\begin{proof}%[ of Lemma~\ref{thm:1}]
    We will use the reduction from \sssattt that was described in Section~\ref{sec:tishard}. This reduction takes polynomial time. 
In particular, given an instance of \sssattt, $\phi$ with $|X|$ variables and $|C|$ clauses, the number of twin pair intervals in the \tis is $|X|+3|C|$, and the domain of $r$ (the horizon) is $[m]$, where $m=16|X|+2|C|+1$, which is polynomial in the number of intervals in the \tis instance. We conclude that the reduction is polynomial in $|X|$ and $|C|$.

    Based on Lemma~\ref{lemma2_SATtoTIS} and  \ref{lemma1_SATtoTIS}, we have established that the \sssattt instance $\phi$ is satisfiable if and only if the constructed \tis instance $\mathcal{I}_\phi$ admits a valid selection. 
\end{proof}

 \subsection{Hardness of \dsprlong}\label{subsec:hardness:dspr}
We present a polynomial-time reduction from \tis to a new problem that we refer to as \dsprlong (\dspr).
Before defining the \dspr problem, we first define the notion of a \emph{dome} as follows.

\begin{definition}
    A \emph{dome} is formed by four positive integers $(a,b,c,d)$ such that $a \leq b < c \leq d$ and $b-a=d-c$. We refer to $(a,d)$ (respectively $(b,c)$) as an \emph{arc} with endpoints $a$ and $d$ (respectively $b$ and $c$). We refer to $(a,d)$ and $(b,c)$ the \emph{outer arc} and \emph{inner arc} of the dome, respectively. When $a=b$ and $c=d$, we call the dome a \emph{singleton dome} and otherwise call it a \emph{regular dome}.
\end{definition}

 In our figures, a singleton dome $(a,b)$ is shown as \includegraphics[scale=0.3]{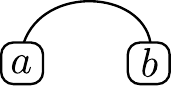} and regular dome $(a,b,c,d)$ is shown as \includegraphics[scale=0.3]{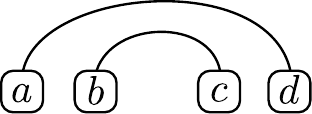}.

\begin{definition}
    An instance $(D,m)$ of the \emph{\dsprlong} (\dspr) is defined by a multiset $D= \{D_1,\ldots, D_n\}$ of domes and a positive integer $m$, where $D_i = (a_i,b_i,c_i,d_i)$ s.t. $d_i \leq m$ for some integer $m$ that we refer to as \emph{horizon}. The decision problem asks whether there is a multiset $S$ of arcs with exactly one arc from each dome $D_i$, such that, for any $t \in [m]$, it holds that $\mathcal{N}_S(t) \leq t$, where $\mathcal{N}_S(t)$ denotes the number of arc endpoints in $S$ with value at most $t$, counting each endpoint as many times as it appears across different arcs in $S$.
\end{definition}

% \figcomponentdomes

\begin{example}\label{example:dspr}
        Consider an instance of the \dspr with domes \linebreak$D= \{(2,3,4,5), (3,4), (4,5,9,10), (7,8)\}$ and $m=10$. A valid solution for this instance is $S= \{(2,5), (3,4), (5,9),(7,8)\}$. For example, when $t=5$, we have $\mathcal{N}_S(t) = 5$, which is no more than $t$. 
        In particular, selected endpoints $\leq t$ are $\{2,5,3,4,5\}$.   
        Note that there are two endpoints with a value of $5$, which belong to two different arcs.
\end{example}

Next, we explain how the \tis problem reduces to the \dspr problem.

\subparagraph*{Construction.}
Given an arbitrary instance $\mathcal{I} = (I,r, m')$ of \tis, where $|I| = n'$ for some positive $n'$, we construct the corresponding instance $\mathcal{D} = (D,m)$ of \dspr as follows. For each twin pair of intervals $(I_i, \overline{I_i}) \in I$, where $I_i = (a', b')$ and $\overline{I_i} = (c', d')$, we construct a regular dome $D_i = (a_i,b_i,c_i,d_i)$ where $a_i = 6n'a', b_i = 6n'b'+3n', c_i = 6n'c',$ and $d_i = 6n'd'+3n'$. 
It is easy to verify that $D_i$ is indeed a dome, that is $b_i-a_i=d_i-c_i$ (see Observation~\ref{obs:indeeddome}).

We define $m$ to be a large enough horizon. In particular, we let $m = 164(n'm')^2$. For any $t \in [m]$ that is a multiple of $6n'$ and $t\leq 3n'(2m'+1)$, we further add $d(t)$ some singleton domes that all start at $t$ and end at $m$. In other words, we add $d(t)$ identical singleton domes $(t,m)$. Next, we explain how $d(t)$ is defined.  For convenience, we let $d(t) = 0$ if $t$ is not a multiple of $6n'$ or $t> 3n'(2m'+1)$. Suppose $t$ is indeed a multiple of $6n'$.
The idea is to define $d(t)$ in a way to project the requirements imposed by the restriction function $r$ in $\mathcal{I}$ to the requirement $\mathcal{N}_S(t)\leq t$ in \dspr.

Let $\mathfrak{D}_i^t$ be the set of regular domes with exactly $i$ endpoints before or at point $t$, and define $c(t) = 2|\mathfrak{D}_4^t| + |\mathfrak{D}_3^t| + |\mathfrak{D}_2^t|$. 
\begin{example}
   For $t = 18$, dome $(1,4,5,8) \in \mathfrak{D}_4^t$, dome $(11,14,16,19)\in \mathfrak{D}_3^t$, dome $(9,10,26,27) \in \mathfrak{D}_2^t$, dome $(17,20,21,24)\in \mathfrak{D}_1^t$, and dome $(25,26,27,28) \in \mathfrak{D}_0^t$. Figure~\ref{fig:dometypes} provides an illustration.
\end{example}

\begin{figure}
	\centering
	\includegraphics[width=0.85\linewidth]{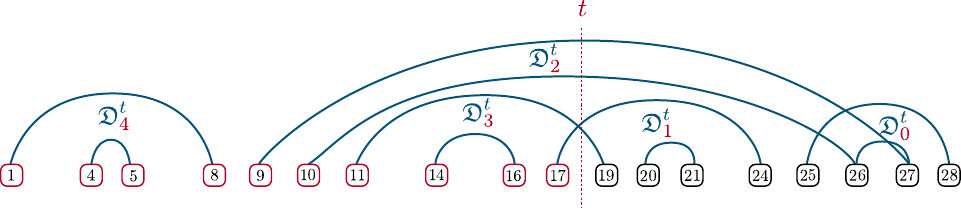}
	\caption{An example of different types of domes based on their positions relative to $t=18$}
	\label{fig:dometypes}
\end{figure}

Intuitively, this definition of $c$ implies that there are $c(t)$ arc endpoints with value at most $t$ in any valid solution $S$ (a set of arcs), 
regardless of the choices made to form $S$. This is because:
\begin{itemize}
    \item (i) all arc endpoints of domes in $\mathfrak{D}^t_4$ are at most $t$, and two of them (associated with one arc) contribute to $c(t)$.
    \item (ii) three arc endpoints of domes in $\mathfrak{D}^t_3$ are at most $t$, and any such dome contributes at least 1 to $c(t)$. \item 
(iii) the left endpoints of both arcs of domes in $\mathfrak{D}^t_2$, and since one arc is in $S$, the dome contributes exactly $1$ to $c(t)$.  
\end{itemize}

Finally, we let 

\scalebox{.9}{
\begin{minipage}{1.2\textwidth}
\begin{align*}
d(t) = t - c(t) - \sum_{j < t} d(j) - r(t/(6n')).
\end{align*}
\vspace{1mm}
\end{minipage}}
The scaling argument that is applied when forming the \dspr instance from \tis instance ensures that $d(t)$ is indeed non-negative for any $t\in [m]$ (see Observatio~\ref{obs:dt_nonnegative} for details).
This completes our construction of the \dspr instance. In a nutshell, we have added one regular dome per twin interval in the \tis instance, and for any $t\in[m]$, we added extra identical singleton domes to capture the requirements imposed by the restriction function in the \tis instance. 

\begin{example}
\label{example::drpsrTis}
    Consider an instance $\mathcal{I} = (I,r,m')$ of the \tis problem with \linebreak $I = \{((1,3)(4,6)),((2,3),(5,6)), ((2,3),(4,5))\}$ and $m'=6$. Suppose $r(1)=r(2)=r(6)=3$,  $r(3)=r(4)=1$, and $r(5)=2$. The corresponding instance $\mathcal{D}=(D,m)$ of \dspr has regular domes $D= \{(18,63,72,117),(36,63,90,117),(36,63,72,99)\}$ and $m=164(n'm')^2=164\cdot(3\cdot6)^2 = 53,136$ (see Figure~\ref{fig:dsprTis}). In addition, $3n'(2m'+1) = 117$ for any $t \in [1,117]$ that is a multiple of $6n'=18$, $d(t)$ singleton domes $(t,m)$ are added. Some examples of $d(t)$ are shown as follows.

    \scalebox{.9}{
    \begin{minipage}{1.2\textwidth}
    \begin{align*}
        &d(18)=18-0-0-3=15&&c(18)=0\\
        &d(36)=36-0-15-3=18&&c(36)=0\\
        &d(54)=54-0-(18+15)-1=20&&c(54)=0\\
        &d(72)=72-3-(15+18+20)-1=15 &&c(72)=3\\
    \end{align*}
    % \vspace{1mm}
    \end{minipage}}
\end{example}

\subparagraph*{Correctness.}
First, we provide some intuitions about the correctness of the reduction. 
We show that any valid solution $S$ for instance $\mathcal{D}$ of \dspr bijects to a valid solution $S'$ of $\mathcal{I}$. 
Note that $S$ contains exactly one arc, the inner or the outer arc, of each regular dome $D_i$ of $\mathcal{D}$ (in addition to singleton domes, which are contained in any valid solution for $\mathcal{D}$). Now, $S'$
selects the left interval $I_i$ when the outer arc of $D_i$ is selected in $S$ and the right interval $\overline{I_i}$ when the inner arc of $D_i$ is selected in $S$ (see Figure~\ref{fig:dsprTis}).
We let $\Gamma_{S'}(t')$ denote the number of intervals selected by $S'$ that include $t'\in [m']$.
We show that $S$ is a valid solution for $\mathcal{D}$ if and only if $S'$ is a valid solution to $\mathcal{I}$.

\begin{figure}
    \centering
    \begin{minipage}[b]{0.42\textwidth}
        \centering
        \includegraphics[width=0.75\linewidth]{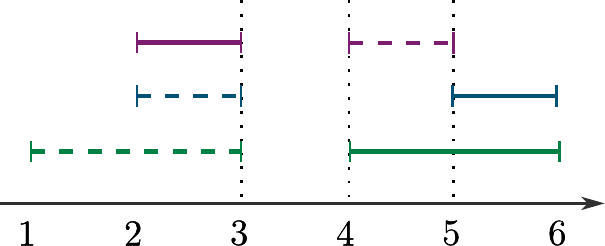} % Image 1
        \subcaption{The instance $\mathcal{I}$ of \tis} % Caption for Image 1
        \label{fig:dsprTis1}
    \end{minipage}
    \hfill
    \begin{minipage}[b]{0.53\textwidth} % Adjust the width for your needs
        \centering
        \includegraphics[width=0.93\linewidth]{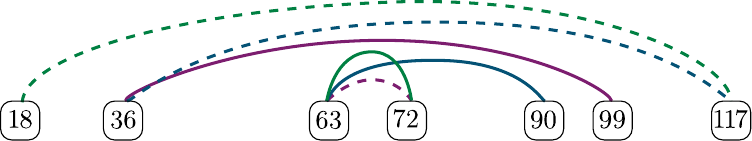} % Image 2
        \subcaption{The corresponding instance $\mathcal{D}$ of \dspr (only regular domes are depicted.) }
        \label{fig:dsprTis2}
    \end{minipage}
    
    \caption{An illustration of the construction of an instance $\mathcal{D}$ of \dspr from an instance $\mathcal{I}$ of \tis, as explained in Example~\ref{example::drpsrTis}. Solid lines show the selected intervals in a valid solution of $\mathcal{I}$ and the selected arcs in the corresponding solution in $\mathcal{D}$.}
    \label{fig:dsprTis}
\end{figure}

Suppose $S$ is a valid solution for \dspr. Then we must have $\mathcal{N}_S(t) \leq t$ for any $t\in[m]$, where $\mathcal{N}_S(t)$ is the number of arc endpoints in $S$ with value at most $t$. As mentioned earlier, $c(t)$ endpoints will be present in any valid solution regardless of the choices to form $S$. Moreover, $S$ must contain all endpoints of singleton domes in $\mathcal{D}$ for $j\leq t$. That is, any valid solution for $\mathcal{D}$, including $S$, must contain $c(t) + \sum_{j\leq t} d(j) $ arc endpoints.  
In addition to these fixed points, there will be more arcs in $S$, which are contributed by the domes in $\mathfrak{D}^t_1$ and $\mathfrak{D}^t_3$ as follows.
    Let $D_i$ be a regular dome in $\mathcal{D}$, and let $(I_i,\overline{I}_i)$ be its corresponding twin intervals in $\mathcal{I}$.
\begin{itemize}
    \item Suppose $D_i\in \mathfrak{D}^t_1$. Now, if  
    $S$ contains the inner 
    arc of $D_i$, the contribution of $D_i$ to $\mathcal{N}_S(t)$ would be $0$. 
    This is equivalent to including the right interval $\overline{I_i}$ in $S'$, and thus 
    the 
    twin interval $(I_i,\overline{I_i})$ does not contribute $\Gamma_{S'}(t')$ for $t'=t/6n$. Informally, the contribution of the dome $D_i$ to the left-hand side of inequality $N_S(t) \leq t$ and the contribution of $(I_i,\overline{I_i})$ to the left-hand side of the inequality $\Gamma_{S'}(t')\leq r(t')$ 
    will be both $0$. 
    Similarly, if $S$ contains the outer arc of $D_i$, the contribution of $D_i$ to $\mathcal{N}_S(t)$ would be $1$.
    This is equivalent to including the left interval $\overline{I_i}$ in the \tis instance, and  $(I_i,\overline{I_i})$ contribute $1$ item to 
    $\Gamma_{S'}(t')$. Informally, the contribution of the dome $D_i$ to the left-hand side of inequality $\mathcal{N}_S(t) \leq t$ and the contribution of $(I_i,\overline{I_i})$ to the left-hand side of 
    $\Gamma_{S'}(t/{(6n')}) \leq r(t')$
    will be both $1$.  

    \item Suppose $D_i\in \mathfrak{D}^t_3$. Assume $S$ contains the inner (respectively the outer) arc of $D_i$. In this case, in addition to the fixed $1$ unit of the contribution of $D_i$ to $\mathcal{N}_S(t)$ (captured in $c(t)$), it further contributes $1$ (respectively $0$) unit to $\mathcal{N}_S(t)$. 
    This is equivalent to including the right interval $\overline{I}_i$ (respectively the left interval $I_i$) in $S'$. 
    In this case, the number of selected intervals in $S'$ that intersect $t$ is increased by $1$ (respectively $0$). Intuitively, both left-hand sides of $\mathcal{N}_S(t) \leq t$ and $\Gamma_{S'}(t/{(6n')}) \leq r(t')$ increase by $1$ (respectively $0$). 
\end{itemize}

We note that $n=|D|$ is polynomial in both $n'$ and $m'$, and thus $m=164(n'm')^2$ is polynomial in $n$. Therefore, we can conclude the following lemma, which establishes the NP-hardness of \dspr. We start the formal proof with the following observation showing that the regular domes in $\mathcal{D}$ are indeed valid. 

\begin{observation}\label{obs:indeeddome}
For any $i\in[n']$, any constructed $D_i = (a_i,b_i,c_i,d_i)$ in the reduction is a valid dome, i.e., $b_i - a_i = d_i - c_i$. 
\end{observation}
\begin{proof}
For each twin pair of intervals $I_i = (a_i',b_i')$ and $\overline{I_i} = (c_i',d_i')$ we know that $b_i'-a_i' = d_i' - c_i'$ by the definition of the \tis problem. Multiplying by $6n'$ and adding $3n'$ to each side, we have $(6n'b_i' + 3n') - 6n'a_i' = (6n'd_i' + 3n') - 6n'c_i'$. Based on the construction of $D_i$, $(6n'b_i' + 3n') - 6n'a_i'=b_i-a_i$ and $(6n'd_i' + 3n') - 6n'c_i'=d_i-c_i$.
\end{proof}

Next, we observe that the number $d(t)$ of extra singleton domes added for any $t\in [m]$ is non-negative.

\begin{observation}\label{obs:dt_nonnegative}
   In the construction of $\mathcal{D}$ from $\mathcal{I}$, for any $t\in [m]$, we have $d(t) \geq 0$. 
\end{observation}
\begin{proof}
    For any $t$ that is not divisible by $6n'$ or $t> 3n'(2m'+1)$, we have $d(t)=0$, and the lemma holds for these values of $t$. For other values of $t$, which are all $6n'$ or larger, we first establish that $d(t)\geq 5n' - c(t)$.  For $t=6n'$, we have $d(t) = 6n' - c(6n') - r(1) \geq 5n'-c(6n')$ (for any value of $t$, we have $r(t)\leq n'$). In what follows, we assume $t\geq 12n'$ and deduce

    \scalebox{0.8}{
    \begin{minipage}{1.2\textwidth}
    \begin{align}
        d(t) &= t-c(t) - \Sigma_{j<t}d(j) - r(t/(6n')) \nonumber\\
        &= t - c(t) - \Sigma_{j\leq t-6n'}d(j) - r(t/(6n')) \label{eq:l2}\\ 
        &= t - c(t) - d({t-6n'}) - \Sigma_{j<t-6n'}d(j) - r(t/(6n')) \nonumber\\
        &= t - c(t) - ( (t-6n') - c(t-6n') - \Sigma_{j<t-6n'}d(j) -r(\frac{t-6n'}{6n'}) ) - \Sigma_{j<t-6n'}d(j) - r(t/(6n')) \label{eq:l4}\\
        &= -c(t) + 6n' +c(t-6n') + r(\frac{t-6n'}{6n'}) - r(t/(6n')).\nonumber
    \end{align} \vspace{1mm}
    \end{minipage}
    }
    
    In Equality \eqref{eq:l2}, we used the fact that singleton domes are added only for multiples of $6n'$. In Equality \eqref{eq:l4}, we used the definition of $d(t-6n')$. 
    As $c({t-6n'})$ and $r(\frac{t-6n'}{6n'})$ are both non-negative, and by definition, $r(t/(6n'))$ is at most $n$, we can conclude that

    \scalebox{.9}{
    \begin{minipage}{1.2\textwidth}
    \begin{align*}
        d(t) \geq -c(t) + 6n' - r(t/(6n')) \geq 5n' - c(t).
    \end{align*}
    \vspace{1mm}
    \end{minipage}}
    Moreover, we can bound $c(t)$ as follows: 
    
    \scalebox{.9}{
    \begin{minipage}{1.2\textwidth}
    \begin{align*}
        c(t) = 2|\mathfrak{D}_4^t| + |\mathfrak{D}_3^t| + |\mathfrak{D}_2^t| \leq 2(|\mathfrak{D}_4^t| + |\mathfrak{D}_3^t| + |\mathfrak{D}_2^t|)\leq 2n'.
    \end{align*}
    \vspace{1mm}
    \end{minipage}}
    Thus, $d(t) \geq 5n' -c(t) \geq 3n' \geq 0$.  
\end{proof}

For each dome $D_i \in D$, let the boolean variable $x_i$ indicate whether the outer arc of $D_i$ is selected in $S$ or not. Specifically, $x_i = 1$ means the outer arc is selected, while $\neg{x_i} = 1$ means the inner arc is selected. 

The following two lemma express the values of $\mathcal{N}_S(t)$ and $\Gamma_{S'}(t)$ in a way that allows us to relate a solution $S$ for the \dspr instance $\mathcal{D}$ to its associated solution $S'$ in the \tis instance $\mathcal{I}$. These lemmas will later help us in proving the correctness of the reduction. 

\begin{lemma}
\label{lemma:tool:TIS2DOME}
Let $\mathcal{D}=(D,m)$ be the \dspr problem constructed from an instance $\mathcal{I}$ of the \tis problem, and let $S$ be a multiset of arcs that contains exactly one arc from each dome $D_i \in D$. For any $t\in[m-1]$, the number of arc endpoints in $S$ with value at most $t$ is exactly $\mathcal{N}_S(t) = \sum_{i \in \mathfrak{D}_1^t} x_i + \sum_{i \in \mathfrak{D}_3^t} \neg{x_i} +|\mathfrak{D}_2^t| +|\mathfrak{D}_3^t| + 2|\mathfrak{D}_4^t| + \sum_{j \leq t} d(j)$.
\end{lemma}
\begin{proof}
Fix a value of $t\in [m-1]$. We consider five cases, depending on the positions of domes relative to $t$, to compute the contribution of different domes to $\mathcal{N}_s(t)$. See Figure~\ref{fig:dometypes} for an illustration. 

\begin{itemize}
    \item Domes with all the endpoints after $t$, namely domes that are in $\mathfrak{D}_0^t$, do not contribute to $\mathcal{N}_S(t)$.

    \item Each dome with exactly $1$ endpoint (the left endpoint of the outer arc) before or on $t$, namely, a dome in $\mathfrak{D}_1^t$, adds exactly one unit to $\mathcal{N}_S(t)$, which happens when the outer arc of $D$ is selected. Thus, the total contribution of domes in $\mathfrak{D}_1^t$ to  $\mathcal{N}_S(t)$ is $\sum_{i \in \mathfrak{D}_1^t} x_i$.
    
    \item For each dome with exactly $2$ endpoints (the left endpoint of each arc) before or on $t$, namely, domes in $\mathfrak{D}_2^t$, selecting either of them adds one endpoint to $\mathcal{N}_S(t)$. Therefore, the contributions of domes in $\mathfrak{D}_2^t$ to $\mathcal{N}_S(t)$ is exactly $|\mathfrak{D}_2^t|$.

    \item Consider a dome $D$ with exactly $3$ endpoints before or on $t$, namely, a dome in $\mathfrak{D}_3^t$. The left endpoint of the arc of $D$ that is present in $S$ contributes $1$ unit to $\mathcal{N}_S(t)$ (regardless of the outer or inner arc being present in $S$). Moreover, if the inner arc of $D$ is selected, there will be another unit of contribution. Therefore, the contributions of  domes in $\mathfrak{D}_3^t$ to $\mathcal{N}_S(t)$ is exactly $|\mathfrak{D}_3^t| + \sum_{i \in \mathfrak{D}_3^t} \neg{x_i}$.

    \item For each dome with all $4$ endpoints before or on $t$, namely domes in $\mathfrak{D}_4^t$, for each of the arcs, its two endpoints are before or on $t$, thus domes in $\mathfrak{D}^t_4$ contribute $2|\mathfrak{D}_4^t|$ endpoints to $\mathcal{N}_S(t)$. 
    
\end{itemize}
Adding the singleton domes added before or at point $t$, which is $d(t)$, we can write $\mathcal{N}_S(t) = \sum_{i \in \mathfrak{D}_1^t} x_i + \sum_{i \in \mathfrak{D}_3^t} \neg{x_i} +|\mathfrak{D}_2^t| +|\mathfrak{D}_3^t| + 2|\mathfrak{D}_4^t| + \sum_{j \leq t} d(j)$.
\end{proof}

\begin{lemma}\label{lem:gamUpBound}
Let $\mathcal{D}=(D,m)$ be the \dspr problem constructed from an instance $\mathcal{I}=(I,r,m')$ of the \tis problem, and let $S'$ be a selection of intervals from $\mathcal{I}$ that contains exactly one interval from each twin interval $(I_i,\overline{I_i}) \in I$. 
For any $t'\in[m-1]$, we have 
$\Gamma_{S'}(t') = \sum_{i \in \mathfrak{D}_3^t} \neg{x_i} +  \sum_{i \in \mathfrak{D}_1^t} x_i  $, where $t=6n't'$.
\end{lemma}

\begin{proof}
    Consider a regular dome $D_i = (a, b, c, d)$ added in $\mathcal{D}$ for the twin intervals $(a', b')$ and $(c', d')$ in $\mathcal{I}$.
For a given $t' \in [m-1]$, let $t = 6t'n$. Recall that  $c = 6n'c'$, and $d = 6n'd' + 3n$ by the construction of $\mathcal{D}$ from $\mathcal{I}$. One can assert that $t'\in [a', b']$ if and only if $D_i \in \mathfrak{D}^t_1$. Therefore, the term $\sum_{i \in \mathfrak{D}_1^t} x_i$ 
  corresponds to the twin intervals in $\mathcal{I}$ whose left interval 
contributes to $\Gamma_{S'}(t')$. 
Similarly, one can assert that $t' \in [c', d']$ if and only if $D_i \in \mathfrak{D}_3^t$. Therefore, the term $\sum_{i \in \mathfrak{D}_3^t} \neg{x_i}$ represents the domes in $\mathfrak{D}_3^t$ for which the inner arc is selected; this corresponds to the twin intervals in $\mathcal{I}$ whose right interval 
contributes to $\Gamma_{S'}(t')$. 
Summing the two terms, we get 
  $\Gamma_{S'}(t') = \sum_{i \in \mathfrak{D}_3^t} \neg{x_i} +  \sum_{i \in \mathfrak{D}_1^t} x_i $.
\end{proof}

Now, we are ready to prove the correctness of the reduction. 

\begin{lemma}
\label{lemma:DSPR2TIS}
If the answer to the instance $\mathcal{D}=(D,m)$ of \dspr is yes, then the answer to its corresponding instance $\mathcal{I}=(I,r,m')$ of \tis is also yes. 
\end{lemma}
\begin{proof}
Consider a solution $S$ that certifies a yes answer for the instance $\mathcal{D}$ of \dspr. We will construct a solution $S'$ for $\mathcal{I}$ from $S$ and prove that it is a valid solution for $\mathcal{I}$. For each regular dome $D_i$ in $\mathcal{D}$  and its corresponding twin intervals $(I_i, \overline{I_i})$ in $\mathcal{I}$, if the inner arc of $D_i$ is selected in $S$, we select $\overline{I_i}$ in $S'$. Similarly, if the outer arc of $D_i$ is selected by $S$, we select $I_i$ in $S'$. 

We need to show that for each $t'\in[m']$, the number of selected intervals in $S'$ intersecting $t'$ is at most $r(t')$. Let $t= 6n't'$. Since $S$ is a valid solution for $\mathcal{D}$, it holds that $\mathcal{N}_S(t)\leq t$. Using Lemmas ~\ref{lemma:tool:TIS2DOME} and~\ref{lem:gamUpBound}, we can write the following. 

\scalebox{.8}{
\begin{minipage}{1.2\textwidth}
    \begin{align}
     \Gamma_{S'}(t')&=\sum_{i \in \mathfrak{D}_3^t} \neg{x_i} +  \sum_{i \in \mathfrak{D}_1^t} x_i && \text{Lemma~\ref{lem:gamUpBound}}\\ 
     & = \mathcal{N}_S(t) - (|\mathfrak{D}_2^t| +|\mathfrak{D}_3^t| + 2|\mathfrak{D}_4^t| + \sum_{j \leq t} d(j)) \label{eq:DSPR2TIS} && \text{Lemma~\ref{lemma:tool:TIS2DOME}}\\
     & \leq t - (|\mathfrak{D}_2^t| +|\mathfrak{D}_3^t| + 2|\mathfrak{D}_4^t| + \sum_{j \leq t} d(j)) && \mathcal{N}_S(t) \leq t \label{eq:DSPR2TIS3}\\
    &\leq t- c(t) - \sum_{j \leq t} d(j)  && \text{by definition of }c(t)\\  
     & = t - c(t) - \sum_{j < t} d(j) - d(t) \nonumber\\ 
     & = t - c(t) - \sum_{j < t} d(j) - (t - c(t) - \sum_{j < t} d(j) - r(t/(6n'))) && \text{by definition of }d(t)\nonumber\\
     & = r(t'). \label{eq:DSPR2TIS2}
\end{align} \vspace{1mm}
\end{minipage}}

Therefore, $\Gamma_{S'}(t') \leq r(t')$ for any $t'\in[m]$ and thus $S'$ is a valid solution for instance $\mathcal{I}$.
\end{proof}

\begin{lemma}
\label{lemma:TIS2DSPR}
If the answer to the instance $\mathcal{I}=(I,r,m')$ of \tis is yes, then the answer to its corresponding instance $\mathcal{D}=(D,m)$ of \dspr is also yes.
\end{lemma}
\begin{proof}
Consider a solution $S'$ that certifies a yes answer for the instance $\mathcal{I}$ of \tis. We will construct a solution $S$ from $S'$ and prove that it is a valid solution for $\mathcal{D}$. For each regular dome $D_i$ in $\mathcal{D}$ and its corresponding twin interval pair $(I_i,\overline{I_i})$ in $\mathcal{I}$, if the left interval $I_i$ is included in $S'$, we include the outer arc in $S$, and if the right interval $\overline{I_i}$ is in $S'$ we include the inner arc in $S$. Clearly, all singleton intervals are also included in $S$. This ensures that exactly one arc from each dome is in $S$. 

In order to show that $S$ is a valid solution for $\mathcal{D}$, we will show that for each $t\in [m]$, $\mathcal{N}_{S}(t) \leq t$. By the construction of $\mathcal{D}$, all singleton domes find their right endpoints at $m$. Let $q=6n'm'+3n'$ and observe that all arcs of regular domes find their endpoint at $\leq q$. We establish $\mathcal{N}_{S}(t) \leq t$ using different arguments depending on the value of $t$ relative to $q$.
\begin{itemize}

\item Suppose $t \in [1,q]$ and $t$ is a multiple of $6n'$. Since $S'$ is a valid solution for $\mathcal{I}$, it must hold that, for each $t'\in[m']$, $\Gamma_{S'}(t')\leq r(t')$.
consider the points that ares Based on Lemma~\ref{lemma:tool:TIS2DOME}, we can rewrite the inequality on point $t = 6n't'$ in the \dspr instance as follows.

\scalebox{.8}{
\begin{minipage}{1.15\textwidth}
    \begin{align}
     \mathcal{N}_S(t) &= \sum_{i \in \mathfrak{D}_3^t} \neg{x_i} +  \sum_{i \in \mathfrak{D}_1^t} x_i  + (|\mathfrak{D}_2^t| +|\mathfrak{D}_3^t| + 2|\mathfrak{D}_4^t| + \sum_{j \leq t} d(j)) && \text{Lemma~\ref{lemma:tool:TIS2DOME}}\label{eq:TIS2DSPR}\\
     &= \Gamma_{S'}(t')  + (c(t) + \sum_{j \leq t} d(j)) &&\text{Lemma~\ref{lem:gamUpBound}}\\
&   
     \leq r(t')+ c(t) + \sum_{j \leq t} d(j)  && \Gamma_{S'}(t')\leq r(t')\\  
     & = r(t') + c(t) + \sum_{j < t} d(j) + d(t) \nonumber\\ 
     & = r(t') + c(t) + \sum_{j < t} d(j) + (t - c(t) - \sum_{j < t} d(j) - r(t/(6n'))) && \text{by definition of }d(t)\nonumber\\
     & = t. \label{eq:TIS2DSPR2}
\end{align} \vspace{1mm}
\end{minipage}}

\item Suppose $t\in[1,q]$ and $t$ is multiple of $3n'$ but not a multiple of $6n'$. 
Let $t_p = t-3n'$ and note that $t_p$ is a multiple of $6n'$. Therefore, we can write (from the previous case) $\mathcal{N}_{S}(t_p) \leq t_p$.
We claim that $\mathcal{N}_{S}(t) \leq \mathcal{N}_{S}(t_p) + 2n'$. 
By Lemma~\ref{lemma:tool:TIS2DOME}, $\mathcal{N}_{S}(t)=\sum_{i \in \mathfrak{D}_3^t} \neg{x_i} +  \sum_{i \in \mathfrak{D}_1^t} x_i  + c(t) + \sum_{j \leq t} d(j)$. Moreover, we have $\sum_{j\leq t} d(j) = \sum_{j\leq t_p} d(j)$ because singleton domes start only in multiples of $6n'$. On the other hand, the sum of the three other terms, namely, $\sum_{i \in \mathfrak{D}_3^t} \neg{x_i} +  \sum_{i \in \mathfrak{D}_1^t} x_i  + c(t)$ is at most $2n'$ for any $t$ (because the number of regular domes is exactly $n'$). Therefore, $\mathcal{N}_S(t)-\mathcal{N}_{S}(t_p) \leq 2n'$ (in an extreme case, when the sum of the three terms is $0$ in $t_s$ and $2n'$ in $t$, the difference between $\mathcal{N}_S(t)$ and $\mathcal{N}_S(t_p)$ will be $2n'$).
We can conclude that $\mathcal{N}_S(t) \leq \mathcal{N}_S(t_p) + 2n' \leq t_p + 2n' \leq t$.

\item Suppose $t\in [1,q]$ and $t$ is not a multiple of $3n'$. Let $t_s$ be the largest multiple of $3n'$ that is smaller than $t$. Since all arcs endpoints are on multiple of $3n'$, we can write $\mathcal{N}_S(t) = \mathcal{N}_S(t_s) \leq t_s < t$. The first inequality holds because we established $\mathcal{N}_S(t_s) \leq t_s$ for values of $t_s$ that are multiples of $3n'$ in previous cases. 

    \item Suppose 
$t \in (q, m)$. Based on our construction, there is no arc endpoint at $t'$ where $q < t' < m$. Thus, in order to show $\mathcal{N}_{S}(t) \leq t$, it suffices to show that $\mathcal{N}_{S}(q)\leq q$, which we will establish in the next case (where $t \leq q)$ 
\item 
Suppose $t=m$. The contributions of each dome to $\mathcal{N}_S(t)$ is exactly 2. Given that we have $n'$ regular domes and $\sum_{t\in[q]}d(t)$ singleton domes, we can write $\mathcal{N}_S(m)=2n'+\sum_{t\in [q]}2d(t)$. Therefore, the following holds.

\scalebox{.9}{
\begin{minipage}{1.2\textwidth}
\begin{align*}
 \mathcal{N}_S(m) &=   2n'+\sum_{t\in [q]}2d(t) \\
 &\leq 2n'+\sum_{t\in [q]}2t && d(t)\leq t\\
    & \leq 2n' + 2q^2\\
    & = 2n'+2(6n'm'+3n')^2 && q=6n'm'+3n'\\
    &\leq 164(n'm')^2 = m.
\end{align*}
\vspace{1mm}
\end{minipage}}

\end{itemize}

 In conclusion, we have $\mathcal{N}_{S}(t)\leq t$ for all values of $t\in[m]$, and we can conclude $S$ certifies that $\mathcal{D}$ is a yes instance of the \dspr problem. 
\end{proof}

From the above lemmas, we can conclude the main result of this section.
\begin{restatable}{lemma}{Domeishard}\label{thm:TIS2DOME}
Answering instances $(D,m)$ of the \dspr problem is NP-hard even if $m$ is polynomial in $|D|$.
\end{restatable}

\begin{proof}
We note that $n=|D|$ is polynomial in both $n'$ and $m'$. Moreover, $m=164(n'm')^2$, which is polynomial in $n'$ and $m'$ (and thus $n$). 
The reduction takes polynomial time in $m$ and $n$ and hence is polynomial in $m'$ and $n'$. Finally, based on Lemma~\ref{lemma:DSPR2TIS} and \ref{lemma:TIS2DSPR}, we have established that the instance $\mathcal{I} = (I,r,m')$ is a yes-instance of \tis if and only if the corresponding instance $\mathcal{D} = (D,m)$ is a yes-instance of the \dspr problem. This completes the proof of equivalence and validates the reduction from \tis to \dspr.
\end{proof}

\subsection{Hardness of \telebg in \Flower Graphs} \label{subsection:telebghard}

We refer to the decision variant of the broadcasting problem as the \telebg problem with instances $(G,s,\rho)$, where $G$ is the input graph, $s$ is the source, and $\rho$ is a positive integer. The decision question asks whether it is possible to complete broadcasting in $G$ from $s$ within $\rho$ rounds. This section presents our final reduction, from \cds to \telebg in \flower graphs.

\subparagraph*{Construction.} 
Given an instance $\mathcal{D}=(D,m)$ of \cds we construct an instance $\mathcal{T}=(G, r, \rho=2m)$ of \telebg, where $G$ is a \flower graph with \pistil $r$. For a given $x \in [\rho]$, let $\tilde{x} = \rho - x$. We construct the graph $G$ by creating a \cater (see Definition~\ref{def:ourflowergraph}) for each dome in $\mathcal{D}$ as follows. 
\begin{itemize}
    \item The \cater of a singleton dome $(a,b)$ is a path of length $2\tilde{a}+2\tilde{b}$ with endpoints $p$ and $q$ (here, $p,q,z$ are special vertices in the \cater, where $z$ is any arbitrary vertex). 
    \item The \cater of a regular dome $(a,b,c,d)$ is formed by a path of length $2\tilde{b}+2\tilde{d}+1$ with endpoints $p$ and $q$, which are two special vertices of the \cater. The other special vertex $z$ of the \cater is the vertex at distance $2\tilde{b}$ of $p$ (and thus distance $2\tilde{d}$ of $q$). In addition to the vertices on the path between $p$ and $q$, the \cater includes $\tilde{a} - \tilde{b}$ (which equals $\tilde{c}-\tilde{d}$) other vertices which find $z$ as their sole neighbor.
    \end{itemize}

\begin{figure}
    \centering
    \begin{minipage}[b]{0.35\textwidth}
        \centering
        \includegraphics[width=.9\linewidth]{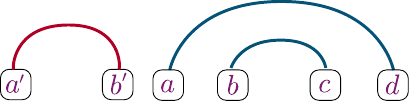} % Image 2
        \subcaption{An instance $\mathcal{D}$ of the \cds instance}% formed by a singleton and a regular dome} % Caption for Image 2
        \label{fig:regular_cdsToB}
        
    \end{minipage}
    \hfill
    \begin{minipage}[b]{0.6\textwidth} % Adjust the width for your needs
        \centering
        \includegraphics[width=.97\linewidth]{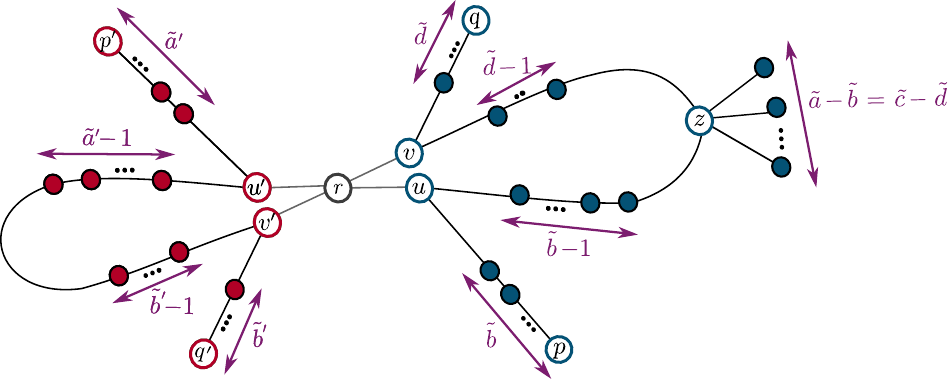} % Image 1
        \subcaption{The corresponding instance $\mathcal{T}=(G,r, \rho)$ of the \telebg constructed from $\mathcal{D}$. } % Caption for Image 1
        \label{fig:cdsToBSubgraphs}
    \end{minipage}
    \caption{An illustration of the construction of the \telebr instance corresponding to a \cds instance}
    \label{fig:domes_cdsToB}
\end{figure}

For any \cater $C$ constructed above (from either a singleton or a regular dome), we label the vertex at distances $\tilde{a}$ form $p$ as the \emph{first gate} of $C$, denoted as $u$, and the one at distance $\tilde{b}$ from $q$ as the \emph{second gate of} $C$ and denote it with $v$. We form the equivalent instance of the \telebg problem by forming a graph $G$ with a \pistil vertex $r$ (source) that is connected to all gates of all reduced caterpillars formed by the domes of $D$. By Definition~\ref{def:ourflowergraph}, $G$ will be a \flower graph with center $r$. Figure~\ref{fig:domes_cdsToB} provides an illustration of this reduction.

\subparagraph*{Correctness.}
Before proving the correctness of our reduction, we prove the following lemmas for broadcasting in the reduced caterpillars associated with singleton and regular domes, respectively. 

\begin{restatable}{lemma}{caterbr}\label{lemma:broudguess:singletondomeiff}
Consider a \cater $C$ of a singleton dome $D=(a,b)$ with gates $u$ and $v$. One can complete broadcasting in $C$ within $\rho$ rounds if and only if $u$ is informed no later than time $a$ and $v$ is informed no later than time $b$.
\end{restatable}
\begin{proof}
    First, suppose the first gate $u$ is informed at time $a'\leq a$ and $v$ is informed at time $b'\leq b$. We explain how broadcasting in $C$ can be completed within $\rho$ rounds. For that,  $u$ (respectively $v$) first informs its neighbor on its path to the special vertex $p$ (respectively $q$) and then its other neighbor on its path towards $v$ (respectively $u$). Given that the distance of $u$ to $p$ is $\tilde{a}$ (respectively the distance of $v$ to $q$ is $\tilde{b}$), all vertices on the path between $u$ and $p$ (respectively $v$ to $q$) will be informed by round $a+\tilde{a} = \rho$ (respectively $b + \tilde{b} = \rho$). See Figure~\ref{fig:singleDome_broad} for an illustration.

    Next, assume broadcasting in $C$ has been completed by $\rho$ rounds. Given that $p$ (respectively $q$) must receive the message through $u$ (respectively $v$), it must be that $u$ (respectively $v$) is informed by round $\rho - \tilde{a} = a$ (respectively $\rho - \tilde{b} = b$).
\end{proof}

\begin{lemma}\label{lemma:broudguess:regulardomeiff}
    Consider a \cater $C$ of a regular dome $D=(a,b,c,d)$ with gates $u$ and $v$. One can complete broadcasting in $C$ within $\rho$ rounds if and only if one of the following happens: 
    \begin{itemize}
        \item  (i) $u$ is informed by time $b$ and $v$ is informed by time $c$. 
        \item (ii) $u$ is informed by time $a$ and $v$ is informed by time $d$.
    \end{itemize}
   
\end{lemma}

\begin{proof}
First, we show that if either (i) or (ii) hold, then broadcasting completes within $\rho$ rounds.

\begin{itemize}
    \item 
Suppose (i) holds, that is, $u$ and $v$ are informed by times $b$ and $c$, respectively. 
We describe a broadcast scheme $S$ that completes within $\rho$ rounds. See Figure~\ref{fig:regular_broadright} for an illustration. 

In $S$, vertex $u$ first informs the neighbor on its path towards $p$, and then it informs the neighbors on its path towards $z$. This means that, by round $\rho$, all $\tilde{b}$ vertices on the path between $u$ and $p$ will be informed. Similarly, $\tilde{b}-1$ vertices on the path between $u$ and $z$ are informed (this excludes $z$). 
On the other hand, $v$ first informs the neighbor on its path towards $z$ and then the neighbor on its path towards $q$. 
This means that, by round $\rho$, all $\tilde{d}$ vertices on the path between $v$ and $q$ will be informed; this is because the number of rounds between the time that $v$ is informed and $\rho$ is $\tilde{c}$, and since $v$ first informs the neighbor towards $z$, $\tilde{c}-1$ rounds remain for informing vertices between $v$ and $q$. This number of rounds is sufficient because the length of the path between $v$ and $q$ is $\tilde{d}\leq \tilde{c}-1$. This way, $z$ will receive the message at time $c+\tilde{d}$, and assuming $z$ informs its uninformed leaves in arbitrary order, the broadcasting in $S$ completes within $c + \tilde{d} + \tilde{c} - \tilde{d} = \rho$ rounds.

\item Next, assume (ii) holds, that is, $u$ and $v$ are respectively informed by rounds $a$ and $d$. See Figure~\ref{fig:regular_broadleft} for an illustration. In this case, $u$ first informs its neighbor on the path towards $z$ and then its other neighbor on the path towards $p$. 
This means that by round $a+\tilde{b} + 1 \leq a+\tilde{a} = \rho$, all vertices on the path between $a$ and $p$ will be informed. Similarly, by round $a$, all the $\tilde{b}$ vertices on the path between $u$ and $z$ (including $z$) will be informed. Thus, all neighbors of $z$ will be informed by $a+\bar{b} + \bar{a}-\bar{b} = \rho$.

\begin{figure}
    \centering
        \includegraphics[width=0.27\linewidth]{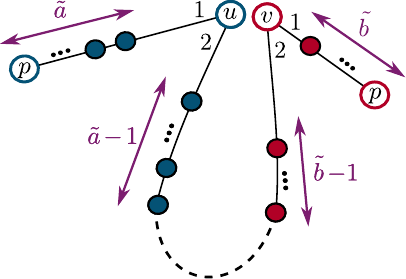} % Image 2
        \caption{
        The broadcast scheme of the \cater 
        associated with a singleton dome $(a,b)$.
        Vertex $u$ (respectively $v$) is responsible for informing the vertices highlighted in blue (respectively red).}
	\label{fig:singleDome_broad}
\end{figure}

On the other hand, $v$ first informs its neighbor on the path toward $q$ and then its neighbor on the path towards $z$ (excluding $z$). Thus, all vertices on the path between $v$ and $q$ will be informed by $d+\tilde{d} = \rho$ while vertices on the path between $v$ and $z$ (excluding $z$) will be informed by time $d + 1 - (\tilde{d}-1) = \rho$. 
\end{itemize}

\begin{figure}
    \centering
    \begin{minipage}[b]{0.47\textwidth}
        \centering
        \includegraphics[width=0.6\linewidth]{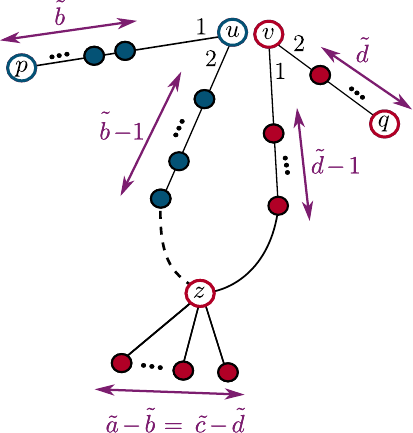} % Image 1
        \subcaption{
        $u$ and $v$ are informed at rounds $b$ and $c$, respectively. This corresponds to selecting arc $(b, c)$ in the associated dome.
        }  % Caption for Image 1
        \label{fig:regular_broadright}
    \end{minipage}
    \hfill
    \begin{minipage}[b]{0.47\textwidth} % Adjust the width for your needs
        \centering
        \includegraphics[width=0.6\linewidth]{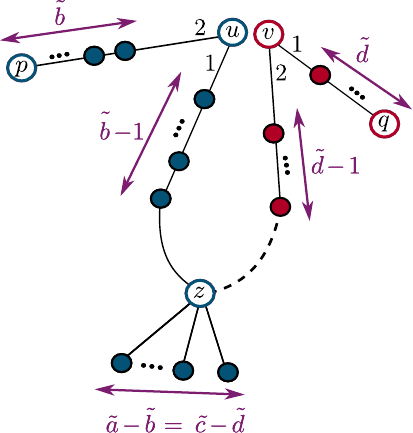} % Image 2
        \subcaption{$u$ and $v$ are informed at rounds $a$ and $d$, respectively. This corresponds to selecting arc $(a,d)$ in the associated dome. } % Caption for Image 2
        \label{fig:regular_broadleft}
    \end{minipage}
    
    \caption{
    Two possibilities for informing vertices of a \cater (associated with a regular dome). The gates are $u$ and $v$, where vertex $u$ (respectively $v$) is responsible for informing the vertices highlighted in blue (respectively red).}
    \label{fig:regular_broadcastcases}
\end{figure}

Next, assume that there is a broadcast scheme $S$ that completes within $\rho$ rounds. We want to show either (i) or (ii) holds. First, note that in the broadcast tree of $S$, vertices $u$ and $v$ cannot be in an ancestor-descendent relationship. Otherwise, if $u$ is an ancestor of $v$ (respectively $v$ is an ancestor of $u$), then $q$ (respectively $p$) receives the message no earlier than round $\tilde{b} + 2\tilde{d} > \rho$
(respectively $2\tilde{b} + \tilde{d} > \rho$). Second, we note that $u$ (respectively $v$) cannot receive the message later than round $b$ (respectively $d$); otherwise $p$ (respectively $q$) will receive the message after round $\rho$. Moreover, since $z$ has $\tilde{a} - \tilde{b}$ (which equals $\tilde{c}-\tilde{d}$) neighbor that all are leaves, it must receive the message by round $\rho - (\tilde{a} - \tilde{b}) = a + \tilde{b}$ (which equals $\rho - (\tilde{c} - \tilde{d}) = c + \tilde{d}$) to complete broadcasting within $\rho$ rounds. This means that either $u$ is informed by round $a$ or $v$ is informed by round $c$ (otherwise, $z$ will be informed too late). We conclude that either (i) $u$ is informed by $a$ and $v$ is informed no later than $d$  or (ii) $u$ is informed by $b$ and $v$ is informed by $c$. 
\end{proof}

We are now ready to prove the main result of this section.

\begin{theorem} \label{thm:hardness-flower}
    \telebr problem is NP-complete for \flower graphs.
\end{theorem}
\begin{proof}
We use the construction mentioned in Section~\ref{subsection:telebghard} to reduce any instance $\mathcal{D}=(D,m)$ of \dspr to an instance $\mathcal{T}=(G,r,\rho)$. 
Note that each dome $D_i$ in $\mathcal{D}$ is bijected to a \cater $C_i$ in $G$. The two possibilities for selecting an arc from a regular dome in $D_i$ translate to two possibilities for informing the gates of $C_i$.

Assume there is a valid solution $S$ for $\mathcal{D}$ in \dspr. We will show there is a broadcasting scheme $S'$ for $\mathcal{T}$ that completes within $\rho$ rounds. For that, we sort all endpoints of all arcs included in $S$ in the non-decreasing order. Consider any dome $D_i \in D$, and let $e_1$ and $e_2$ be the endpoints of the arc selected from $D_i$ in $S$. Assume $e_1$ and $e_2$ have ranks $i_1$ and $i_2$ in the sorted order. In the broadcast scheme $S'$, the source $r$ informs the gates of \cater $C_i$ associated with $D_i$ at rounds $i_1$ and $i_2$. Given that ranks of all arc endpoints in $S$ are distinct, $r$ informs at most one neighbor at each given round. 

Next, we show that $S'$ completes within round $\rho$.
Given that $S$ is a valid solution for $\mathcal{D}$, for any $t\in m$, we have $\mathcal{N}_S(t) \leq t$, where $\mathcal{N}_S(t)$ is the number of arc endpoints in $S$ that are at most $t$. On the other hand, for any $e$ that is the endpoint of an arc in $S$, we have $\rrank(e) \leq \mathcal{N}_S(e)$ (be definition, $\rrank(e)$ is upper bounded by the number of arc endpoints in $S$ that are at most $e$, while $\mathcal{N}_S(e)$ is exactly the number of endpoints that are at most $e$). We conclude that $\rrank(e) \leq e$. 
Therefore, the gates of any \cater $C_i$ in $G$ associated with a singleton dome $D_i = (a,b)$ in $D$ receive the message in $S'$ by rounds $(a,b)$. Similarly, the gates of any \cater
$C_i$ associated with a regular
dome $D_i = (a,b,c,d)$ receive the message in $S'$ by rounds $(a,d)$ (if the arc $(a,d) \in S$) or by rounds $(b,c)$ (if the arc $(b,c)$ is in $S$). Therefore, by Lemma~\ref{lemma:broudguess:singletondomeiff} and \ref{lemma:broudguess:regulardomeiff}, broadcasting in $S'$ completes by round $\rho$.

For the other side of the reduction, suppose there is a broadcast scheme $S'$ for $\mathcal{T}$ that completes within $\rho$ rounds. We will explain how to construct a valid solution $S$ for $\mathcal{D}$. 
Consider any \cater $C_i$ in $G$ associated with a dome $D_i \in D$. 
Let $(\alpha_i,\beta_i)$ be the rounds that \pistil $r$ informs the gateways of $C_i$.
Now, if $D_i = (a,b)$ is a singleton dome, by Lemma~\ref{lemma:broudguess:singletondomeiff}, it must be that the gates of $C_i$ are informed by rounds $(a,b)$, that is $\alpha \leq a$ and $\beta \leq b$. In this case, the single arc of $D_i$ is included in $S$. 
Next, if $D_i = (a,b,c,d)$ is a regular dome, by Lemma~\ref{lemma:broudguess:regulardomeiff}, it must be that the gates of $C_i$ are informed by either rounds $(a,d)$ or $(b,c)$, that is either ($\alpha \leq a$ and $\beta \leq d$) or ($\alpha\leq b$ and $\beta \leq c$) must hold. 
We will include arc $(a,d)$ in $S$ in the former case and $(b,c)$ in the latter case. 
In other words, any arc endpoint $x$ in $S$ is associated with a gate that is informed within round $x$ in $S'$. 
To show that $S$ is a valid solution for $\mathcal{D}$, we show that for any $t\in [m]$, we have $\mathcal{N}_S(t) \leq t$. 
Consider the set of endpoints in $D$ that contribute to $\mathcal{N}_S(t)$. As mentioned above, any of these endpoints is associated with a gate that is informed by $r$ within round $t$. Given that $r$ informs up to $t$ gates by round $t$, it must be that $\mathcal{N}_S(t) \leq t$. 
We conclude that $S$ is a valid instance of $\mathcal{D}$. 

One can verify that the size of the graph $G$ in the \telebg instance is polynomial on $n=|D|$ and $m$. Therefore, given that the \dspr is NP-hard by  Lemma~\ref{thm:TIS2DOME}, we conclude that \telebg is NP-hard. \telebr in general graphs is known to be in NP \cite{slater1981nptree}. Together, these results establish the NP-completeness of \telebg in \flower graphs.
\end{proof}

\section{Constant-Factor Approximation for Bounded Pathwidth}
\label{sec:pathwidth-approx}

Graphs of pathwidth $1$ are caterpillars~\cite{ProskurowskiT99}, which are special types of trees, for which the \telebr problem is solvable in linear time \cite{fraigniaud2002polynomial}. On the other hand, for graphs with a pathwidth larger than $1$, the \telebr problem is NP-hard, as established by our result for the hardness of the problem in \flower graphs (Theorem~\ref{thm:hardness-flower}). Recall that \flower graphs have pathwidth 2 (Observation~\ref{obs:flower-pathwidth}). 

In this section, we establish the existence of a constant-factor approximation for \telebr on graphs with bounded pathwidth. Recall that we use $\bropt(G,s)$
to denote the optimal broadcasting time for an instance $(G,s)$. We will demonstrate that the algorithm of~\citet{elkin2006sublogarithmic}, which has an approximation factor of $\oh\left(\frac{\log n}{\log \bropt(G, s)}\right)$ for any graph $G$, achieves a constant factor approximation for graphs of bounded pathwidth. For general graphs, this algorithm has an approximation factor of $\oh(\log n/ \log \log n)$ (because $\bropt(G,s)\geq  \log n$), which is the best known approximation factor.  

For any graph $G$ of pathwidth $\pwidth$, we will show that $\bropt(G,s) = \Omega(n^{4^{-(\pwidth+1)}})$, which establishes that the algorithm of~\citet{elkin2006sublogarithmic} has a constant factor approximation for graphs of bounded pathwidth $\pwidth$.

To find a lower bound for $\bropt(G,s)$ where $G$ is a graph of constant pathwidth, we repeatedly remove a vertex from $G$ and show that the broadcasting in the remainder of $G$ is not much slower compared to $G$. For that, we will use the following lemma, which holds for any instance of \telebr (but we only use it for graphs of bounded pathwidth). 

\begin{restatable}{lemma}{constantBRminus}\label{lemma:brminus}
Consider any instance $(G,s)$ of \telebr, and let $v$ be any arbitrary vertex in $G$. Let $H_1,\ldots, H_m$ be the connected components resulting from removing $v$ from $G$ ($m\geq 1$). 
For any $i \in [m]$, choose an arbitrary vertex $s_i\in H_i$.
Then, the following inequality holds: $\sum_{i\in[m]} \bropt(H_i, s_i) \leq \bropt(G,s)(2\bropt(G,s) + 1)$.
\end{restatable}

% \begin{proof}[Sketch]
% Suppose $T^*$ is an optimal broadcast tree for $(G,S)$. Removing $v$ from $T^*$ creates disjoint subtrees $\tau_i$. We ``glue'' these trees arbitrarily to get a spanning tree for each connected component $H_i$ and show that these trees provide the broadcast guarantees stated in the lemma (details in Appendix~\ref{appendix-hardness}).
% \end{proof}

\begin{proof}
Consider the optimal broadcast tree $T^*$ that completes broadcasting in $(G,s)$ in $\bropt(G,s)$ rounds. Suppose $v$ has $d$ children $u_1, \dots, u_d$ in $T^*$ for some $d \geq 0$. Removing $v$ from $T^*$ results in $d+1$ subtrees, $\tau_0, \dots, \tau_d$, where $\tau_0$ is rooted at $s$ and $\tau_i$ is rooted at $u_i$ for each $i \in [1,d]$. Note that all vertices in any fixed $\tau_j$ belong to the same connected component $H_i$.
Let $\tau^i_1, \dots, \tau^i_{p_i}$ be the subtrees that form the connected component $H_i$ for some $p_i \geq 0$. The indexing is defined in a way to ensure that $s_i \in \tau^i_1$, and there is an \emph{auxiliary edge} between a vertex in $\tau^i_{j}$ and a vertex in $\tau^i_{j'}$ for some $j' \leq j-1$ (this is possible because $H_i$ is connected). Figure~\ref{fig:designatedNodesExample} provides an illustration.

     \begin{figure}
	\centering
	\includegraphics[width=0.35\linewidth]{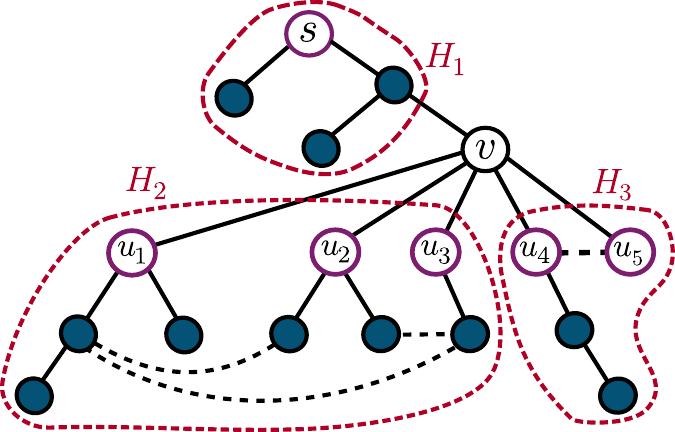}
	\caption{An illustration of the proof of Lemma~\ref{lemma:brminus}. Auxiliary edges are dotted. 
    In this example, the connected components $H_1$ to $H_3$ formed after removing $v$ are shown. 
    }
    
    \label{fig:designatedNodesExample}
    \end{figure}
    
A broadcasting scheme for $H_i$ can be formed as follows:
\begin{itemize}
\item The message is transmitted from $s_i$ to the root of $\tau^i_1$ within $\bropt(G,s)$ rounds.
\item Broadcasting within $\tau^i_1$ completes in at most $\bropt(G,s)$ rounds, in a similar way that the message is broadcasted in $T^*$.
\item The message is sent to the next subtree $\tau^i_2$ along an auxiliary edge in one extra round.
\end{itemize}

The above process takes at most $2\bropt(G,s) +1$ rounds. 
Similarly, each subtree $\tau^i_j$ takes at most $2\bropt(G,s)+1$ rounds to receive the message and broadcast it within $\tau^i_j$. Using this procedure iteratively, broadcasting within $(H_i, s_i)$ completes in at most $p_i(2\bropt(G,s)+1)$ rounds.
Summing over all $H_i$, we get

\scalebox{.9}{
\begin{minipage}{1.2\textwidth}
\begin{align*}
    \sum_{i \in[m]} \bropt(H_i,s_i) = (2\bropt(G,s)+1) \sum_{i \in[m]} p_i.
\end{align*}
\vspace{1mm}
\end{minipage}}

Note that $\sum_{i\in[m]} p_i$ is a lower bound for $\bropt(G,s)$ since it is the degree of $v$ in $T^*$. Then we conclude

\scalebox{.9}{
\begin{minipage}{1.2\textwidth}
\begin{align*}
    \sum_{i \in[m]} \bropt(H_i,s_i) \leq \bropt(G,s)(2\bropt(G,s)+1).
\end{align*}
% \vspace{1mm}
\end{minipage}}
\end{proof}

For completeness, we show that Lemma~\ref{lemma:brminus} is asymptotically tight. 

\begin{observation}\label{observation:tighti}
    There are instances of the \telebr problem $(G,s)$ and vertex $v \in G$ for which Lemma~\ref{lemma:brminus} is asymptotically tight.
\end{observation}

\begin{proof}
    We provide $(G,s)$ and $v\in G$ such that if removing $v$ from $G$ results in connected components $H_1,\ldots, H_m$, then $\sum_{i\in [m]} \bropt(H_i,s_i) = \Omega( \bropt(G,s)^2)$. Let $G$ be a \emph{fan graph} formed by a path $P$ of $n-1$ vertices, where each vertex $i \in [1,n-1]$ is connected to a \emph{center} vertex $c$. Suppose $v=c$ and $s$ is an endpoint of $P$. 
    We claim that $\bropt(G, s) = \oh(\sqrt{n})$. 
    This broadcast time can be achieved by a broadcast tree in which $v$ has $\Theta(\sqrt{n})$ children, and any neighbor of $v$ in the tree is responsible for informing $\Theta(\sqrt{n})$ other vertices (see Figure~\ref{fig:fangraph-example}). On the other hand, $G \setminus \{v\}$ has only one connected component, which is a path $H_1$ in which, setting the source as $u_1$ gives $\sum_{i\in[m]} \bropt(G_i,s_i) = \bropt(H_1,u_1) = n-1 \in \Omega(\bropt(G,s)^2)$. 
\end{proof}

\begin{figure}
	\centering
\includegraphics[width=0.35\linewidth]{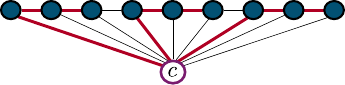}
	\caption{An illustration of Observation~\ref{observation:tighti}}
    \label{fig:fangraph-example}
\end{figure}

The main result of this section can be stated as follows.

\begin{restatable}{theorem}{constanttheorem}\label{thm:approximation-constant}
There is a polynomial-time algorithm for \telebr in graphs with constant pathwidth $\pwidth$, achieving an approximation ratio of $\oh(4^{\pwidth})$.
\end{restatable}

% \begin{proof}[Sketch]
\paragraph{Proof Overview.}
As mentioned above, it suffices to prove a lower bound of $\Omega(n^{4^{-(w+1)}})$ for $\bropt(G,s)$. 
For that, we assume $G$ is a ``$\pwidth$-path'' in the sense that any two vertices that share a bag are neighbors. If $G$ is not a $\pwidth$-path, we will add the missing edges to get a $\pwidth$-path without increasing the pathwidth.
 Clearly, the addition of new edges does not decrease the broadcast time of $G$, and since we are looking for a lower bound for the broadcast time, it suffices to focus on $\pwidth$-paths.

We assume a \emph{standard path decomposition} of the input graph $G$, in which no bag is a subset of another. In such decompositions, we define the notion of \emph{span} of a vertex $v$ as the number of bags where $v$ appears. The span of a decomposition is then the maximum span over the spans of all its vertices. 

The proof of lemma is established by induction over the size of the input graph $G$. For that, we consider two cases. First, if the span of the decomposition is ``small'', we argue that the diameter of $G$ will be ``large'', and then the desired lower bound for $\bropt(G,s)$ holds.
On the other hand, when the span is ``large'', we consider the vertex $v_m$ that has the maximum span and  
consider the graph $G_{v_m}$ induced by vertices that share a bag with $v_m$. Note that $G_{v_m}$ is a smaller graph compared to $G$, and we can show that broadcasting in $G_{v_m}$, starting from any vertex, cannot be much slower than broadcasting in $G$, starting from $s$ (Lemma~\ref{lemma:br:inducedHG}). 
Moreover, we will use Lemma~\ref{lemma:brminus} to show that if we extract $v_m$ from $G_{v_m}$ to get a set $\{H_1,\ldots, H_q\}$ of disjoint connected components, total broadcast time in these component components (starting from arbitrary sources) is not much slower in the absence of $v_m$. On the other hand, since $v_m$ is removed from $G_{v_m}$, any of these components $H_i$ has a pathwidth that is at least one unit less than $G_{v_m}$, and we can use an inductive argument to achieve a lower bound on the broadcast times of any $H_i$ (Lemma~\ref{lemma:brlbsonleft}). In summary, by the induction hypothesis, total broadcast time in $H_i$'s gives a lower bound for broadcasting in $G_{v_m}$ which itself gives a lower bound for broadcasting in $G$. 
% \end{proof}

Now for the formal proof, we prove a lower bound for the broadcast time of any connected graph $G$ with $n$ vertices and pathwidth $\pwidth \in \oh(1)$. In particular, for any $s\in G$, we will show $\bropt(G,s) \in \Omega(n^{4^{-\pwidth}})$.
Consider a fixed path decomposition of $G$; we assume the bags in the deposition are arranged from left to right. For now, suppose $s$ is in the left-most bag (this assumption will be relaxed later in Lemma~\ref{lemma:brlbsonleft}). 

We assume the path decomposition is ``standard'' in the sense that no bag is a subset of another bag (otherwise, one can remove the smaller bags to attain a standard decomposition). Finally, we assume that $G$ is a ``$\pwidth$-path'' in the sense that any two vertices located in the same bag are connected in $G$ (all edges allowed in the decomposition are present). If the input graph is not a $\pwidth$-path, we can make it $\pwidth$-path by adding all the missing edges. Clearly, the addition of new edges does not decrease the broadcast time of $G$, and since we are looking for a lower bound for the broadcast time, it suffices to focus on $\pwidth$-paths. 

We now proceed with a formal proof.

\begin{lemma}\label{lemma:br:inducedHG}
    Let $(H,s)$ be an instance of the broadcast problem, where $H$ is a connected $\pwidth$-path. Suppose we have a standard path decomposition $D$ of width $\pwidth$ for $H$ where $s$ appears in the first bag of $D$. Let $G$ be an induced subgraph of $H$ formed by vertices that appear in a consecutive set of bags in $D$ and suppose a vertex $s'\in G$ appears in all such bags. Then, we can write $\bropt(G,s') \leq \bropt(H,s)+2\pwidth$.
\end{lemma}
    
\begin{proof}
 Suppose $T_H$ is an optimal broadcast scheme for the instance $(H,s)$. We will construct a broadcast scheme $T_G$ for $(G,s')$ as follows. See Figure~\ref{fig:decomposition_bag_graphs} for an illustration. Update $T_H$ by removing every vertex that is not a part of $G$. Also, if the parent of $s'$ in $T_H$ is also in $G$, remove the edge between them in the updated tree.
 This results in a forest of at most $2\pwidth+1$ subtrees of $T_H$ that contains all vertices in $G$. 
 This is because the roots of these subtrees can only be in the first and last bag of $G$, and at most $2\pwidth+1$ vertices in $G$ are located in these bags (note that $s'$ is in both of the bags). 
 Add edges from $s'$ to the roots of all the subtrees. It is possible to add such edges because $s'$ is connected to every vertex in $G$ as it appears in all bags of $G$, and also, all vertices that share a bag are connected in $H$ (because $H$ is a $\pwidth$-path).
 The result would be a connected tree $T_G$, spanning $G$ and rooted at $s'$. 
 
We show broadcasting in $T_G$, starting from $s'$, does not take more than $2\pwidth$ extra rounds than broadcasting in $T_H$, that is, $\bropt(G,s') \leq \bropt(H,s)+2w$. This holds because $s'$ can inform all roots of the subtrees first within  $2\pwidth$ rounds. The remainder of broadcasting is conducted similarly in $T_H$ and $T_G$. 
\end{proof}

\begin{figure}
    \centering
    \begin{minipage}[b]{0.5\textwidth}
        \centering
        \includegraphics[width=\linewidth]{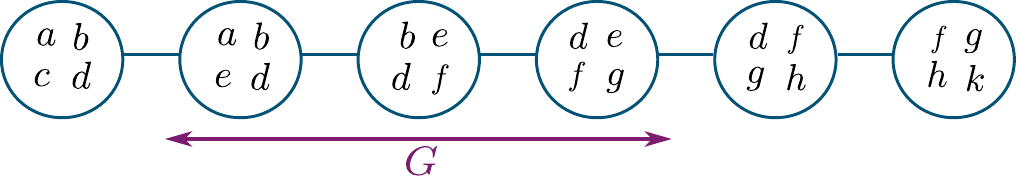} % Image 2
        \subcaption{Standard path decomposition for a graph $H$} % Caption for Image 2
        \label{fig:decomposition_bags}
        
    \end{minipage}
    \hfill
    \begin{minipage}[b]{0.2\textwidth} % Adjust the width for your needs
        \centering
        \includegraphics[width=\linewidth]{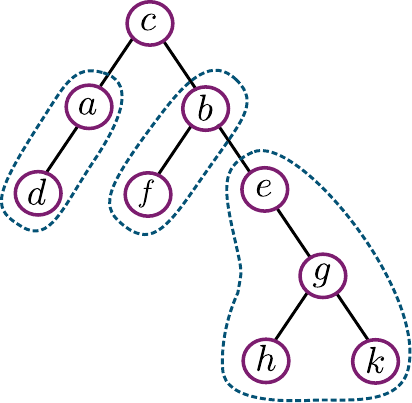} % Image 1
        \subcaption{A broadcast scheme $T_H$ for $(H,s)$} % Caption for Image 1
        \label{fig:decomposition_graph}
    \end{minipage}
    \hfill
    \begin{minipage}[b]{0.25\textwidth} % Adjust the width for your needs
        \centering
        \includegraphics[width=.5\linewidth]{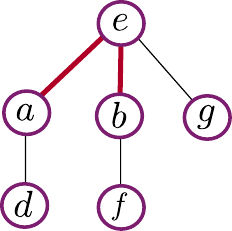} % Image 1
        \subcaption{The associated broadcast scheme $T_G$ for $(G,s')$} % Caption for Image 1
        \label{fig:decomposition_graph2}
    \end{minipage}
    \caption{An illustration of Lemma \ref{lemma:br:inducedHG}}
    \label{fig:decomposition_bag_graphs}
\end{figure}

To establish a lower bound for broadcasting in any connected $\pwidth$-path graph $G$, we prove a slightly stronger result. Suppose we have $k$ connected $\pwidth$-path graphs $H_1, H_2,\ldots, H_k$, each $H_i$ having at least two vertices, a standard path decomposition of width  $\pwidth$ and a vertex $s_i \in H_i$ located at the left-most bag in its decomposition. We prove a lower bound for the total time for broadcasting in these graphs as follows. This is a generalization of our desired lower bound for the case of $k=1$.

\begin{lemma} \label{lemma:brlbsonleft}
    Consider a set of $k\geq 1$ vertex disjoint connected $\pwidth$-path graphs $\mathcal{H} = \{H_1, \ldots, H_k\}$. Suppose any $H_i$ has a standard path decomposition $D_i$ with $\ell_i$ bags, where a source vertex $s_i$ is located in the leftmost bag of $D_i$. Let $\bigell = \sum_{i=1}^k \ell_i$. 
    Then we have
    $\sum_{i=1}^k \bropt(H_i, s_i) \geq f(\bigell, \pwidth)$, where $f(\bigell, \pwidth) = 27^{-\pwidth}(\pwidth!)^{-2}\bigell^{4^{-\pwidth}}$.
\end{lemma}

\begin{proof}
    Let $n$ denote the total size (number of vertices) in all $H_i$'s. We use an induction on $n$ to prove the lemma. In the base of induction, we consider the cases where $n\leq 3\pwidth$, which implies $\bigell \leq 3\pwidth$. Therefore, by the definition for $f(\bigell, \pwidth)$, we have $f(\bigell, \pwidth)\leq 1$. 
    On the other hand, we have $\sum_{i=1}^k \bropt(H_i, s_i) \geq 1$; this is because $H_1$ is formed by at least two vertices, and its broadcast time is at least $1$. Thus, $f(\bigell,\pwidth) \leq \bropt(H_1,s_1) \leq \sum_{i=1}^k \bropt(H_i,s_i)$, and the base of induction holds. 

    In the induction step, we consider two cases, namely $k=1$ and $k>1$.
First, suppose $k>1$. Thus, each $H_i$ has a size $|H_i|< n$. 
    Therefore, we can apply the induction hypothesis on a set of graphs formed by only one $H_i$ (for each $i\in [k]$). This gives $\bropt(H_i,s_i) \geq f(\ell_i,\pwidth)$. 
    Summing over all $i$'s and observing that the function $f(\bigell, \pwidth)$ is sublinear in $\bigell$, we can conclude

\scalebox{.9}{
    \begin{minipage}{1.2\textwidth}
        \begin{align*}
            \sum_{i=1}^k \bropt(H_i, s_i) \geq \sum_{i=1}^k f(\ell_i,\pwidth) \geq f(\bigell,\pwidth).
        \end{align*}
        \vspace{1mm}
\end{minipage}}
    Next, suppose $k=1$. Consider the bags in the standard path decomposition $D_1$ containing any vertex $v$. By the definition of path decomposition, these bags form a consecutive set of bags in $D_1$; we refer to this set as the \textbf{span} of $v$ and denote it by $\sspan(v)$. Let $v_m$ be the vertex with maximum span, i.e., $v_m = \arg \max_v |\sspan(v)|$, and let $M=|\sspan(v_{m})|-1$. For example, in Figure~\ref{fig:decomposition_bags}, we have $span(a) = 2$, $v_m = d$, and $M=3$. We analyze two cases based on the value of $M$. 

    \subparagraph*{Case $\mathbf{1}$: $\mathbf{M}$ is small.}  Suppose $M \leq \frac{\bigell - 1}{f(\bigell, \pwidth)}$ and let $u$ denote any vertex that only appears in the rightmost bag of $D_1$ (since $D_1$ is a standard path decomposition, at least one such vertex $u$ exists). We claim the distance between the source $s$ and $u$ is large enough to establish the desired lower stated in the lemma. Let the shortest path between $ s$ and $u$ be $\langle s(=u_0), u_1,\ldots, u_d(=u)\rangle$. For any $j\in[d]$ given that the span of $u_j$ is at most $M$ and $u_j$ is connected to $u_{j+1}$, it must be that $u_{j+1}$ appears for the first time in a bag within distance $M$ (in $D_1$) of the bag where $u_j$ was first introduced. It means the total number $\bigell$ of bags cannot be more than $dM+1$ or $d\geq (\bigell-1)/M$. Given that the distance between $s$ and $u$ in $H_1$ is an upper bound for the broadcast time, we can conclude 

    \scalebox{.9}{
\begin{minipage}{1.2\textwidth}
    \begin{align*}
        \sum_{i=1}^k \bropt(H_i, s_i)  = \bropt(H_1, s_1)\geq (\bigell-1)/M \geq f(\bigell, \pwidth).
    \end{align*}
    \vspace{1mm}
\end{minipage}}
    
    The last inequality holds because of the assumption that $M\leq \frac{\bigell - 1}{f(\bigell,\pwidth)}$.

    \subparagraph*{Case $\mathbf{2}$: $\mathbf{M}$ is large.} Suppose $M > \frac{\bigell - 1}{f(\bigell, \pwidth)}$. Let $B_s$ and $B_e$ be the first and the last bag of $D_1$ in which $v_m$ appears, respectively. 
   Consider the subgraph $G$ induced by all vertices that appear in any of the 
   bags between $B_s$ and $B_e$ (inclusive of all vertices in $B_s$ and $B_e$). 
    Remove from $G$ the vertex $v_m$; the result would be a graph $G' = \{G'_1,G'_2,\ldots, G'_r\}$, which is a set of  $r(\geq 1)$ connected components which are possibly singletons (a graph with one vertex). Let $c$ denote the number of such singletons ($c\geq 0$) and consider a graph resulting from removing these $c$ singletons from $G'$. We denote the result as $\mathcal{H'}=\{H'_1, H'_2, \ldots, H'_{r-c}\}$, which is a set of vertex-disjoint connected $\pwidth'$-paths. 
    Each $H'_i$ has a standard path decomposition $D'_i$ with $\ell'_i$ bags. Also, we consider an arbitrary $s'_i$ for each $H'_i$ that appears in the first bag of $D'_i$. 
    
    We have $\pwidth' \leq \pwidth-1$ because $v_m$ is removed from $G$ to form $G'$, and it has been present in all bags in the path decomposition of $G'$. Next, we consider two cases as follows based on the value of $c$.

\begin{itemize}
    \item 
    First, we establish the lemma when $c \geq f(\bigell, \pwidth) + 2$.
    We note that each singleton component, except the two that may appear in $B_s$ and $B_e$, finds $v_m$ as their sole neighbor in $H_1$ (this holds because all vertices that share a bag in $D_1$ are connected). Therefore, the degree of $v_m$ in any broadcast tree of $H_1$ is at least $c-2 \geq f(\bigell,\pwidth)$. We conclude that $    \sum_{i=1}^k \bropt(H_i, s_i)  = \bropt(H_1,s_1) \geq f(\bigell, \pwidth)$, and the lemma holds as desired. 

   % ***********************
   \item Next, we consider the case $c < f(\bigell,\pwidth)+2$. Since $v_m$ appears in all bags from $B_s$ to $B_e$ in $D_1$, we can apply Lemma~\ref{lemma:br:inducedHG} 
   to conclude that 
   
   \scalebox{.9}{
\begin{minipage}{1.05\textwidth}
   \begin{align}
       \bropt(G,v_m) \leq \bropt(H_1,s_1) + 2\pwidth. \label{eq:gtoh}
   \end{align}
   \vspace{1mm}
\end{minipage}}

   Therefore, if we can show $\bropt(G,v_m) \geq f(\bigell,\pwidth)+2\pwidth$,
   it follows that $\sum_{i=1}^k \bropt(H_i, s_i)  = 
        \bropt(H_1, s_1)\geq \bropt(G,v_m) -2w \geq f(\bigell, \pwidth)$,
   which will establish the statement of the lemma as desired. 
   
   It remains to prove $\bropt(G,v_m) \geq f(\bigell,\pwidth)+2\pwidth$, for which we will use the induction hypothesis and Lemma~\ref{lemma:brminus} as follows. 
    The number of the vertices in $\mathcal{H}'$ is at most $n-1$; this is because $\mathcal{H}'$ misses $v$ (along with possibly other vertices) compared to $G$, and the size of $G$ is at most $n$. 
    Let $\bigell'=\sum_i\ell'_i$ be the total number of bags in $\mathcal{H}'$. 
    We have $\bigell'=M-c$ because every $H'_i \in \mathcal{H}'$ contains a consecutive set of bags of $G$, and $c$ singleton components, each of which located in exactly one bag, are excluded. Therefore, we can use the induction hypothesis to write
    
    \scalebox{.9}{
\begin{minipage}{1.2\textwidth}
    \begin{align*}
        \sum_{i} \bropt(H'_i, s'_i) 
        \geq \sum_{i} f(\ell_i', \pwidth - 1)
        \geq f(\bigell', \pwidth - 1).
    \end{align*}
    \vspace{1mm}
\end{minipage}}
    The last step is because function $f(\ell',\pwidth)$ is sublinear in $\ell'$ and $\sum_i \ell'_i = L'$.
    Moreover, if we apply Lemma~\ref{lemma:brminus} on $G$, when $v_m$ is removed from $G$, we can conclude 

\scalebox{.9}{
\begin{minipage}{1.05\textwidth}
\begin{align}
    \bropt(G, v_m)(2\bropt(G, v_m) + 1) \geq \sum_{j\in[p+c]}\bropt(G'_j,s'_j) = \sum_{i\in[r-c]}\bropt(H'_i,s'_i).
    \label{eq:br-removed}
\end{align}
\vspace{1mm}
\end{minipage}}

The last inequality holds because broadcasting in each of the singleton graphs that are removed from $G'$ to form $\mathcal{H'}$ takes $0$ round. Therefore, we can write
    
    \scalebox{.9}{
\begin{minipage}{1.2\textwidth}
    \begin{align*}
        3\bropt(G, v_m)^2 \geq& \bropt(G, v_m)(2\bropt(G, v_m) + 1) && \bropt (G, v_m) \geq 1\\
         =& \sum_{i\in[r-c]}\bropt(H'_i,s'_i) && \text{from~\eqref{eq:br-removed}} \\
        \implies \bropt(G, v_m) \geq& \big(\sum_{i \in [r-c]} \bropt(H'_i, s'_i)/3\big)^{1/2} \\
        \geq& \big(f(\bigell', \pwidth - 1)/3\big)^{1/2}. && \text{by the induction hypothesis}
    \end{align*}
    \vspace{1mm}
\end{minipage}}
We claim that $\bigell'\geq \sqrt{\bigell}$. Assuming this claim holds, continuing from the above inequality,
we can write

\scalebox{.9}{
\begin{minipage}{1.2\textwidth}
    \begin{align*}
        \bropt(G, v_m) \geq& \big(f(\bigell', \pwidth - 1)/3\big)^{1/2}\geq \big(f(\sqrt{\bigell}, \pwidth - 1)/3\big)^{1/2}\\
        \geq& \big(27^{-\pwidth+1}\bigell^{4^{-\pwidth+1}/2}((\pwidth-1)!)^{-2}/3\big)^{1/2} && \text{by definition of }f(.)\\
        \geq& 3(27)^{-\pwidth}\bigell^{4^{-\pwidth}}(\pwidth!)^{-2}\pwidth \\
        \geq& 3\pwidth f(\bigell, \pwidth) && \text{by definition of }f(.)\\ 
        \geq& f(\bigell, \pwidth) + 2\pwidth. && \text{because }f(L,\pwidth)\geq 1
    \end{align*}
    \vspace{1mm}
\end{minipage}}
    Recall that $\bropt(H_1,s_1) \geq \bropt(G,v_m) - 2\pwidth$ by Inequality \eqref{eq:gtoh} and combining this with the above inequality, we can conclude $\sum_{i=1}^k \bropt(H_i, s_i)  = 
 \bropt(H_1,s_1)\geq f(\ell,\pwidth)$. Now, it remains to prove that $\bigell'\geq \sqrt{\bigell}$.

    \scalebox{.9}{
\begin{minipage}{1.2\textwidth}
    \begin{align*}
        \bigell' & = M-c \\ 
        &> M-f(\bigell,\pwidth) -2 && c \text{ is assumed less than }  f(\bigell,\pwidth) +2 \\ 
        & >  \frac{\bigell - 1}{f(\bigell, \pwidth)} - f(\bigell, \pwidth) -2.  && M \text{ is assumed larger than } \frac{\bigell - 1}{f(\bigell, \pwidth)} 
        \end{align*}
        \vspace{1mm}
\end{minipage}}
    Therefore, to prove $L'\geq \sqrt{L}$ it suffices to prove $\frac{\bigell - 1}{f(\bigell, \pwidth)} - f(\bigell, \pwidth) -2 \geq \sqrt{\bigell}$, which we establish as follows.

    \scalebox{.9}{
\begin{minipage}{1.2\textwidth}
    \begin{align*}
    &\frac{\bigell - 1}{f(\bigell, \pwidth)} - f(\bigell, \pwidth) -2 \geq \sqrt{\bigell}\\ 
        \iff& \bigell - 1 \geq \sqrt{\bigell}f(\bigell, \pwidth) + f(\bigell, \pwidth)^2 + 2f(\bigell, \pwidth)\\
        \Longleftarrow & \bigell - 1 \geq 2\bigell^{3/4}/27 + \bigell^{1/2}/27^2 + 2\bigell^{1/4}/27 && \text{by definition of } f(\bigell, \pwidth)\\
        \Longleftarrow & \bigell - 1 \geq \bigell^{3/4}/3. 
    \end{align*}
    \vspace{1mm}
\end{minipage}}
    The last inequality is true for any $\bigell \geq 2$. 
    \qedhere
\end{itemize}
\end{proof}

The following lemma generalizes Lemma~\ref{lemma:brlbsonleft} (applied with $k=1$) to the case where the source is not necessarily located in the first bag of the path decomposition.

\begin{lemma}
\label{lemma:pwidth-ell}
   Let $G$ be any connected graph 
   and 
   let $s$ be any vertex of $G$. Suppose $G$ has a standard path decomposition of width $\pwidth$ formed by $\bigell$ bags. Then, we have $\bropt(G,s) \geq f(\bigell-1,\pwidth+1)$, where $f(\bigell, \pwidth) = 27^{-\pwidth}(\pwidth!)^{-2}\bigell^{4^{-\pwidth}}$.
\end{lemma}
\begin{proof}
Let $D = \langle B_1,\ldots, B_\bigell\rangle$ be the standard path decomposition of $G$ of width $\pwidth$, and $B_k$ be the leftmost bag that contains $s$ in $D$.
Add $s$ to any bag $B_x$ of $D$ for $x\in[k-1]$. The result would be a valid path decomposition $D'=\langle B'_1,\ldots, B'_{\bigell'}\rangle$ of width $\pwidth'$, which is not larger than $\pwidth+1$. To keep the path decomposition in the standard form, we may need to merge $B_k$ with $B_{k-1}$. 
Regardless, $\bigell'$ is at least $\bigell-1$. 

We then add edges between $s$ and all vertices of $G$ that appear in a bag $B_i$ where $i\leq k$; the result would be a new graph $G^+$ that is a $(\pwidth+1)$-path (recall that adding edges will not increase the broadcast time). Now, we have a \telebr instance $(G^+,s)$, where $G^+$ is a $(\pwidth+1)$-path and the source is located on $B_1$. Therefore, we can apply   
Lemma~\ref{lemma:brlbsonleft} on a single component $H_1=G^+$ ($k=1$), to conclude we have $\bropt(G,s) \geq \bropt(G^+,s) \geq f(\bigell',\pwidth')\geq f(\bigell-1,\pwidth+1)$. \end{proof}

We are now ready to prove the main result of this section.

\begin{proof}[of Theorem~\ref{thm:approximation-constant}]
Let $(G,s)$ be an instance of \telebr, where $G$ is a graph of pathwidth $\pwidth$. Consider a standard path decomposition $D$ of $G$, where each bag contains at most $\pwidth + 1$ vertices. Let $L$ be the number of bags in $D$ and note that $\bigell \geq \frac{n}{\pwidth + 1}$. This holds because each vertex has to be present in at least one bag. Combining this inequality with Lemma~\ref{lemma:pwidth-ell}, we will obtain

    \scalebox{.9}{
\begin{minipage}{1.06\textwidth}
    \begin{align}
    \label{bronpwidth}
        \bropt(G,s) \geq f(\bigell-1,\pwidth+1) = 27^{-(\pwidth+1)}((\pwidth+1)!)^{-2}\Big(\frac{n}{2(\pwidth + 1)}-1\Big)^{4^{-(\pwidth+1)}}.
    \end{align}
    \vspace{1mm}
\end{minipage}}

    Therefore, assuming $\pwidth = \oh(1)$, we can write $\bropt(G,s)=\Omega(n^{4^{-(\pwidth+1)}})$. The approximation ratio of the algorithm introduced by \citet{elkin2006sublogarithmic} is given by $\oh\left(\log n / \log \bropt(G, s)\right)$. Using Inequality \eqref{bronpwidth}, the approximation factor would be

\scalebox{.9}{
\begin{minipage}{1.2\textwidth}
        \begin{align*}
            \oh\Big(
        \frac{\log n}{
        \log n^{4^{-(\pwidth+1)}}} \Big)
        = \oh(4^\pwidth),
        \end{align*} 
        \vspace{1mm}
\end{minipage}}
        which completes the proof. 
\end{proof}

\section{Concluding Remarks}\label{section:concluding}

In this paper, we resolved an open problem by proving the NP-completeness of \telebr for cactus graphs as well as graphs of pathwidth $2$. We also established a $2$-approximation algorithm for cactus graphs and a constant-factor approximation algorithm for graphs of bounded pathwidth. A possible direction for future work is to improve the approximation factor for graphs of bounded pathwidth or cactus graphs. In particular, it remains an open question whether Polynomial Time Approximation Schemes (PTASs) exist for these graph classes. A major open problem in this domain is determining whether a constant-factor approximation exists for general graphs. While progress on this question has been slow, the algorithmic ideas developed in this work may be applicable to broadcasting in other families of sparse graphs.

\section*{Acknowledgement}
We acknowledge the support of the Natural Sciences and Engineering Research Council of Canada (NSERC)[funding reference number DGECR-2018-00059].
% \newpage

\bibliography{main}

\begin{thebibliography}{27}
\providecommand{\natexlab}[1]{#1}
\providecommand{\url}[1]{\texttt{#1}}
\expandafter\ifx\csname urlstyle\endcsname\relax
  \providecommand{\doi}[1]{doi: #1}\else
  \providecommand{\doi}{doi: \begingroup \urlstyle{rm}\Url}\fi

\bibitem[Ben-Moshe et~al.(2012)Ben-Moshe, Dvir, Segal, and Tamir]{ben2012centdian}
Boaz Ben-Moshe, Amit Dvir, Michael Segal, and Arie Tamir.
\newblock Centdian computation in cactus graphs.
\newblock \emph{Journal of Graph Algorithms and Applications}, 16\penalty0 (2):\penalty0 199--224, 2012.

\bibitem[Bhabak and Harutyunyan(2015)]{BhabakH15}
Puspal Bhabak and Hovhannes~A. Harutyunyan.
\newblock Constant approximation for broadcasting in k-cycle graph.
\newblock In \emph{Proceedings of first Conference of Algorithms and Discrete Applied Mathematics}, volume 8959, pages 21--32, 2015.

\bibitem[{\v{C}}evnik and {\v{Z}}erovnik(2017)]{vcevnik2017broadcasting}
Maja {\v{C}}evnik and Janez {\v{Z}}erovnik.
\newblock Broadcasting on cactus graphs.
\newblock \emph{Journal of Combinatorial optimization}, 33:\penalty0 292--316, 2017.

\bibitem[Damaschke(2024)]{damaschke2024starclique}
Peter Damaschke.
\newblock A linear-time optimal broadcasting algorithm in stars of cliques.
\newblock \emph{Journal of Graph Algorithms and Applications}, 28\penalty0 (1):\penalty0 385--388, 2024.

\bibitem[Ehresmann(2021)]{ehresmann2021approximation}
Anne-Laure Ehresmann.
\newblock Approximation algorithms for broadcasting in flower graphs.
\newblock Master's thesis, Concordia University, 2021.

\bibitem[Elkin and Kortsarz(2002)]{elkin2002lowerbound}
Michael Elkin and Guy Kortsarz.
\newblock Combinatorial logarithmic approximation algorithm for directed telephone broadcast problem.
\newblock In \emph{Proceedings of the thirty-fourth annual ACM symposium on Theory of computing}, pages 438--447, 2002.

\bibitem[Elkin and Kortsarz(2006)]{elkin2006sublogarithmic}
Michael Elkin and Guy Kortsarz.
\newblock Sublogarithmic approximation for telephone multicast.
\newblock \emph{Journal of Computer and System Sciences}, 72\penalty0 (4):\penalty0 648--659, 2006.

\bibitem[Fraigniaud and Mitjana(2002)]{fraigniaud2002polynomial}
Fraigniaud and Mitjana.
\newblock Polynomial-time algorithms for minimum-time broadcast in trees.
\newblock \emph{Theory of Computing Systems}, 35\penalty0 (6):\penalty0 641--665, 2002.

\bibitem[Garey and Johnson(1979)]{GareyJ79}
M.~R. Garey and David~S. Johnson.
\newblock \emph{Computers and Intractability: {A} Guide to the Theory of NP-Completeness}.
\newblock 1979.

\bibitem[Gholami et~al.(2023)Gholami, Harutyunyan, and Maraachlian]{gholami2023fullytree}
Saber Gholami, Hovhannes~A Harutyunyan, and Edward Maraachlian.
\newblock Optimal broadcasting in fully connected trees.
\newblock \emph{Journal of Interconnection Networks}, 23\penalty0 (01):\penalty0 2150037, 2023.

\bibitem[Grigni and Peleg(1991)]{grigni1991tightkbroad}
Michelangelo Grigni and David Peleg.
\newblock Tight bounds on minimum broadcast networks.
\newblock \emph{SIAM Journal on Discrete Mathematics}, 4\penalty0 (2):\penalty0 207--222, 1991.

\bibitem[Harutyunyan and Maraachlian(2007)]{Harutyunyan2007unicyclic}
Hovhannes Harutyunyan and Edward Maraachlian.
\newblock Linear algorithm for broadcasting in unicyclic graphs.
\newblock In \emph{Proceedings of the 13th Annual International Conference on Computing and Combinatorics}, pages 372--382, 2007.

\bibitem[Harutyunyan et~al.(2009)Harutyunyan, Laza, and Maraachlian]{harutyunyan2009necklace}
Hovhannes Harutyunyan, George Laza, and Edward Maraachlian.
\newblock Broadcasting in necklace graphs.
\newblock In \emph{Proceedings of the 2nd Canadian Conference on Computer Science and Software Engineering}, pages 253--256, 2009.

\bibitem[Harutyunyan and Hovhannisyan(2023)]{harutyunyan2023kpath}
Hovhannes~A Harutyunyan and Narek Hovhannisyan.
\newblock Improved approximation for broadcasting in k-path graphs.
\newblock In \emph{International Conference on Combinatorial Optimization and Applications}, pages 111--122, 2023.

\bibitem[Harutyunyan and Liestman(2001)]{harutyunyan2001improved}
Hovhannes~A Harutyunyan and Arthur~L Liestman.
\newblock Improved upper and lower bounds for k-broadcasting.
\newblock \emph{Networks: An International Journal}, 37\penalty0 (2):\penalty0 94--101, 2001.

\bibitem[Harutyunyan and Maraachlian(2008)]{harutyunyan2008unicyclic}
Hovhannes~A Harutyunyan and Edward Maraachlian.
\newblock On broadcasting in unicyclic graphs.
\newblock \emph{Journal of combinatorial optimization}, 16:\penalty0 307--322, 2008.

\bibitem[Harutyunyan et~al.(2023)Harutyunyan, Hovhannisyan, and Maraachlian]{harutyunyan2023chainring}
Hovhannes~A Harutyunyan, Narek Hovhannisyan, and Edward Maraachlian.
\newblock Broadcasting in chains of rings.
\newblock In \emph{Proceedings of the Fourteenth International Conference on Ubiquitous and Future Networks}, pages 506--511, 2023.

\bibitem[Hedetniemi et~al.(1988)Hedetniemi, Hedetniemi, and Liestman]{hedetniemi1988broadsurvey}
Sandra~M Hedetniemi, Stephen~T Hedetniemi, and Arthur~L Liestman.
\newblock A survey of gossiping and broadcasting in communication networks.
\newblock \emph{Networks}, 18\penalty0 (4):\penalty0 319--349, 1988.

\bibitem[Kortsarz and Peleg(1995)]{kortsarz1995approximation}
Guy Kortsarz and David Peleg.
\newblock Approximation algorithms for minimum-time broadcast.
\newblock \emph{SIAM Journal on Discrete Mathematics}, 8\penalty0 (3):\penalty0 401--427, 1995.

\bibitem[Liestman and Peters(1988)]{Liestman1988Boundeddegree}
Arthur~L. Liestman and Joseph~G. Peters.
\newblock Broadcast networks of bounded degree.
\newblock \emph{SIAM Journal on Discrete Mathematics}, 1\penalty0 (4):\penalty0 531--540, 1988.

\bibitem[Proskurowski and Telle(1999)]{ProskurowskiT99}
Andrzej Proskurowski and Jan~Arne Telle.
\newblock Classes of graphs with restricted interval models.
\newblock \emph{Discret. Math. Theor. Comput. Sci.}, 3\penalty0 (4):\penalty0 167--176, 1999.

\bibitem[Robertson and Seymour(1983)]{robertson1983pathwidth}
Neil Robertson and Paul~D Seymour.
\newblock Graph minors. i. excluding a forest.
\newblock \emph{Journal of Combinatorial Theory, Series B}, 35\penalty0 (1):\penalty0 39--61, 1983.

\bibitem[Shang et~al.(2010)Shang, Wan, and Hu]{shang2010unitdisk}
Weiping Shang, Pengjun Wan, and Xiaodong Hu.
\newblock Approximation algorithms for minimum broadcast schedule problem in wireless sensor networks.
\newblock \emph{Frontiers of Mathematics in China}, 5:\penalty0 75--87, 2010.

\bibitem[Slater et~al.(1981)Slater, Cockayne, and Hedetniemi]{slater1981nptree}
Peter~J. Slater, Ernest~J. Cockayne, and Stephen~T. Hedetniemi.
\newblock Information dissemination in trees.
\newblock \emph{SIAM Journal on Computing}, 10\penalty0 (4):\penalty0 692--701, 1981.

\bibitem[St{\"o}hr(1991)]{stohr1991butterfly}
Elena St{\"o}hr.
\newblock Broadcasting in the butterfly network.
\newblock \emph{Information Processing Letters}, 39\penalty0 (1):\penalty0 41--43, 1991.

\bibitem[Tale(2024)]{tale2024double}
Prafullkumar Tale.
\newblock Double exponential lower bound for telephone broadcast.
\newblock \emph{arXiv preprint arXiv:2403.03501}, 2024.

\bibitem[Tovey(1984)]{tovey1984simplified}
Craig~A Tovey.
\newblock A simplified np-complete satisfiability problem.
\newblock \emph{Discrete applied mathematics}, 8\penalty0 (1):\penalty0 85--89, 1984.

\end{thebibliography}
\bibliographystyle{plainnat}
% \appendix
% \input{src/Appendix/A-Appendix-2approx}
% \input{src/Appendix/B-Appendix-hardness}
% \input{src/Appendix/C-Appendix_constant}

% \todo{a task}

% In paper \cite{vcevnik2017broadcasting}, they studied broadcass... .

% \begin{figure}
% 	\centering
% 	\includegraphics[scale=0.8]{figs/h5sample11.pdf}
% 	\caption{my caption
% 	} 
% 	\label{fig:h5-sample-11}
% \end{figure}

% An example of this proof is shown in Figure \ref{fig:h5-sample-11}. 

% \begin{definition}\label{def1}
    
% \end{definition}

% Based on Definition \ref{def1}, ...

% \begin{lemma}
    
% \end{lemma}
% \begin{proof}
    
% \end{proof}

\end{document}